\documentclass[pdflatex,sn-mathphys-num]{sn-jnl}

\usepackage{graphicx}%
\usepackage{multirow}%
\usepackage{amsmath,amssymb,amsfonts}%
\usepackage{amsthm}%
\usepackage{mathrsfs}%
\usepackage[title]{appendix}%
\usepackage{xcolor}%
\usepackage{textcomp}%
\usepackage{manyfoot}%
\usepackage{booktabs}%
\usepackage{algorithm}%
\usepackage{algorithmicx}%
\usepackage{algpseudocode}%
\usepackage{listings}%
\usepackage{bbm}

\theoremstyle{plain}
\newtheorem{thm}{Theorem}[section]
\newtheorem{lem}[thm]{Lemma}

\newtheorem{cor}[thm]{Corollary}

\theoremstyle{definition}
\newtheorem{defn}{Definition}[section]

\newtheorem*{thm*}{Theorem}
\theoremstyle{remark}
\newtheorem*{rem}{Remark}

\usepackage{xcolor}
\usepackage{caption}
\usepackage{subcaption}

\begin{document}

\title[Article Title]{Robust Hedging of path-dependent options using a \textit{min-max} algorithm}


\author[1]{\fnm{Purba} \sur{Banerjee}}\email{purbab@iisc.ac.in}
\equalcont{These authors contributed equally to this work.}

\author[1]{\fnm{Srikanth} \sur{Iyer}}\email{skiyer@iisc.ac.in}
\equalcont{These authors contributed equally to this work.}

\author*[2]{\fnm{Shashi } \sur{Jain}}\email{shashijain@iisc.ac.in}
\equalcont{These authors contributed equally to this work.}

\affil[1]{\orgdiv{Department of Mathematics}, \orgname{Indian Institute of Science}, \orgaddress{\city{Bangalore}, \postcode{560012},  \country{India}}}

\affil*[2]{\orgdiv{Department of Management Studies}, \orgname{Indian Institute of Science}, \orgaddress{\city{Bangalore}, \postcode{560012},  \country{India}}}


\abstract{We consider an investor who wants to hedge a path-dependent option with maturity $T$ using a static hedging portfolio using cash, the underlying, and vanilla put/call options on the same underlying with maturity $ t_1$, where $0 < t_1 < T$. We propose a model-free approach to construct such a portfolio. The framework is inspired by the \textit{primal-dual} Martingale Optimal Transport (MOT) problem, which was pioneered by \cite{beiglbock2013model}. The optimization problem is to determine the portfolio composition that minimizes the expected worst-case hedging error at $t_1$ (that coincides with the maturity of the options that are used in the hedging portfolio). The worst-case scenario corresponds to the distribution that yields the worst possible hedging performance. This formulation leads to a \textit{min-max} problem. We provide a numerical scheme for solving this problem when a finite number of vanilla option prices are available. Numerical results on the hedging performance of this model-free approach when the option prices are generated using a \textit{Black-Scholes} and a \textit{Merton Jump diffusion} model are presented. We also provide theoretical bounds on the hedging error at $T$, the maturity of the target option. }     




\maketitle


\section{Introduction}\label{sec:Introduction}
Derivative pricing and hedging form a crucial part of the financial world. A fundamental step in derivative pricing involves modeling the underlying stock price process under certain predetermined assumptions. A common practice is to calibrate the parameters of the chosen stock price model to the market prices of actively traded vanilla options. Since the choice of the model for the calibration process is not necessarily unique,  under unexpected market scenarios, the predicted stock prices under the chosen model can deviate significantly from the true stock price. This will also result in non-unique prices for the same derivative security, depending on the choice of the model. Hence, from an investor's perspective, hedging their position in the derivative is necessary to protect against unpredictable price fluctuations. An investor can choose a dynamic or a static hedging approach to construct a hedging portfolio. The advantage of not requiring constant monitoring of the market fluctuations and incurring transaction costs each time the underlying hedging portfolio is rebalanced for a dynamic strategy makes static hedging attractive.

In \cite{carr2014static}, the authors obtain a static hedging portfolio of \textit{European} call options with maturity $0<t_1$ to hedge a \textit{European} call option with a longer maturity $T$ when the stock price process satisfies a one-factor Markovian dynamics. In \cite{banerjeemultiperiod}, the authors extend this work to multiple time points. This static hedging approach yields a superior performance than a standard delta hedging approach for a jump-diffusion process like a Merton Jump Diffusion model, where a sudden jump in the stock price process in between the rebalancing times can lead to a considerable gap between the target option with maturity $T$ and the delta hedging portfolio. However, under stochastic volatility models, this static hedging approach performs poorly. We consider a natural extension of this static hedging problem to the robust model-free setting. 

To provide a general overview of our problem, we consider an investor who holds a short position in an option with pay-off $c(X, Y)$ at maturity $T$ where $X$ and $Y$ denote the underlying stock price process at times $t_1,T$ with $0<t_1<T$. The investor wishes to construct a portfolio comprising cash, stocks, and vanilla put/call options on the same underlying asset with maturity $t_1$ to hedge this position. We denote this hedging portfolio as $h(X,w)$ where $w$ denotes the weights of the components, i.e., options, stocks, and cash. A negative value of a component of $w$ denotes a short position in the corresponding asset (options, stocks, or cash) and a positive value denotes a long position. We assume the availability of market prices of the traded call and put options corresponding to different maturities. If the call/put prices are available for all strikes over $[0,\infty)$, a result of \cite{breeden1978prices} then allows for recovery of the marginal distribution of the stock price and is independent of any underlying model assumption. The resulting marginals will be consistent with the available call/put option prices. Then, under the no-arbitrage condition, any pricing measure will be a martingale measure with these marginals. Let $\mathcal{M}(\mu,\nu)$ denote the set of probability measures $\mathbb{P}$ with marginals at times $t_1, T$ given by $\mu, \nu$ respectively and satisfying the martingale condition $\mathbb{E}_{\mathbb{P}}[Y|X]=X$. The marginals $(\mu,\nu)$ denote the true marginal distributions of the underlying stock price process. For simplicity, we assume the rate of borrowing/lending to be zero. Now, if the investor wants to compute the weights with respect to which the worst-case hedging error at maturity $T$ is minimized, then the objective function reduces to
\begin{align}
\label{generalized_objective}
    \begin{split}
       P^{H_T}(\mu,\nu)&:=\inf_{w}\sup_{\mathbb{P}\in\mathcal{M}(\mu,\nu)}\mathbb{E}_{\mathbb{P}}\bigg[\bigg|c(X,Y)-h(X,w)\bigg|\bigg].
    \end{split}
\end{align}
 The inner maximization problem in ($\ref{generalized_objective}$) is a \textit{martingale optimal transport} (MOT) problem.

From an investor's perspective, it is more important to compute the weights corresponding to the worst-case hedging error at $t_1$ (the maturity of the options in the hedging portfolio). The objective function in this case is
\begin{align}
\label{generalized_objective_at_t1}
    \begin{split}
       P^{H_{t_1}}(\mu,\nu)&:=\inf_{w}\sup_{\mathbb{P}\in\mathcal{M}(\mu,\nu)}\mathbb{E}_{\mu}\bigg[\bigg|\mathbb{E}_{\mathbb{P}}[c(X,Y)|X]-h(X,w)\bigg|\bigg]. 
    \end{split}
\end{align}
The inner maximization problem in (\ref{generalized_objective_at_t1}) can be considered as a modified MOT problem, which we denote by Mod-MOT .

The MOT problem for obtaining robust bounds on option prices was pioneered by \cite{beiglbock2013model} and followed by the works of \cite{galichon2014stochastic},\cite{dolinsky2014martingale},\cite{beiglbock2017complete},\cite{hou2018robust},\cite{lutkebohmert2019tightening},\cite{guo2019computational},\cite{bauerle2019martingale},\cite{eckstein2021robust},\cite{bauerle2021consistent} to name a few. 

 Given  the risk-neutral marginal probability distributions $\mu,\nu$ on $\mathbb{R}$ and a measurable cost function $c:\mathbb{R}^2\rightarrow\mathbb{R}$, in a classical optimal transport problem, the objective is to maximize (minimize) 
\begin{align}
     \label{OT}
        \int_{\mathbb{R}^2}c(x,y)\mathbb{P}(dx,dy).
\end{align}
The optimization is over all probability measures $\mathbb{P}$, under the constraints that the marginals of $\mathbb{P}$ are predefined distributions $\mu,\nu$ satisfying 
\begin{align}
    \label{marginal}
    \mathbb{P}(E\times \mathbb{R})=\mu(E)\hspace{0.1cm} \text{and}\hspace{0.1cm} \mathbb{P}(\mathbb{R}\times E)=\nu(E).
\end{align}
 
For a detailed overview of the study of optimal transportation problems, we refer the reader to \cite{villani2021topics}. If we require that $\mathbb{P}$ satisfy an additional martingale constraint \begin{align}
\label{MOT}
    \int_{\mathbb{R}}y\mathbb{P}(dy|x)=x,
\end{align} where $\mathbb{P}(dy|x)$ denotes the conditional distribution of the random variable $Y$ given $X$, then  (\ref{OT})-(\ref{MOT}) is termed the \textit{martingale optimal transport} (MOT) problem. This solution corresponds to an upper (lower) price bound for an option with payoff $c$.

In practice, we have actively traded call prices available only at a finite number of strikes. It is then possible to find discrete measures consistent with observed call prices. When the measures $\mu$ and $\nu$ are discrete, i.e., $\mu(dx)=\sum_{i=1}^{m}\alpha_i\delta_{x_i}(dx)$ and $\nu(dy)=\sum_{j=1}^{n}\beta_j\delta_{y_j}(dy)$, the MOT problem (\ref{OT})-(\ref{MOT}) reduces to a \textit{linear programming} (LP) problem. The LP problem is given by

\begin{align}
\label{MOT_discrete}
\begin{split}
   & \max_{(p_{i,j})_{\in \mathbb{R}_{+}^{mn}}}\sum_{i=1}^{m}\sum_{j=1}^{n}p_{i,j}c(x_i,y_j)\hspace{0.2cm}\text{subject to}\\
   &\sum_{j=1}^{n}p_{i,j}=\alpha_i,
    \sum_{i=1}^{m}p_{i,j}=\beta_j,
    \sum_{j=1}^{n}p_{i,j}y_j=\alpha_i x_i,\hspace{0.2cm}\text{for}\hspace{0.2cm} i=1,..,m;\hspace{0.2cm} j=1,..,n.
\end{split}
\end{align}
One can utilize the iterative Bregman projection to solve the LP, as shown in \cite{benamou2015iterative}.

An important observation is that the maximization problem in the objective function (\ref{generalized_objective}) also reduces to an LP problem for discretely supported marginals $(\mu,\nu)$ and the problem (\ref{generalized_objective}) can be viewed as a \textit{min-max} problem. For the maximization problem in (\ref{generalized_objective_at_t1}), one needs to make certain modifications to transform it into a linear problem for discretely supported marginals $(\mu,\nu)$.
 This is explained in greater detail in Section \ref{Section:Hedging problem}. \cite{davis2014arbitrage} pioneered this LP approach for an MOT problem where a finite number of expectation constraints were provided instead of the marginal constraint $\nu$. For a convex reward function, this yields optimizers with finite support.
 
In a real-world scenario, the true marginals $(\mu,\nu)$ of the underlying stock price process at times $0<t_1<T$ are unknown. Given that only finitely many call option prices are available, one cannot directly use the results in \cite{breeden1978prices} to obtain the true marginal distributions consistent with the call option prices. An alternative approach would be to approximate the solution of the original MOT problem (\ref{OT})-(\ref{MOT}) using an  LP problem of the form (\ref{MOT_discrete}) for discretely supported marginal distributions which are consistent with the available call prices. We need to ensure that the solution of the MOT problem obtained by solving the LP is close to the value obtained from the true underlying as call prices become available over an increasingly dense set of strikes. In \cite{bauerle2021consistent}, the authors prove that if the option’s payoff function $c$ is directionally convex, the optimization over all discrete measures reduces to those with marginals as described below. These measures stochastically dominate all other discrete measures consistent with the observed call prices. A major drawback is that the result is restricted to cases where the true marginals are discrete with compact support. Since the true underlying distribution is unknown, assuming it to have compact support would not be ideal. An extension of their discretization scheme for unbounded measures and relevant convergence results are provided in \cite{schmithals2019contributions}, which form a basis for our work. For unbounded measures, a different approach is provided in \cite{guo2019computational} where the authors introduce an $\epsilon-$ relaxation approach to obtain a sequence of relaxed MOT problems. Their main result provides conditions to ensure the convergence of a given sequence of relaxed MOT problems with discrete marginals to the actual MOT problem (\ref{OT})-(\ref{MOT}) with continuous marginals. They prove the result under some moment assumptions on the true measure. The choice of $\epsilon$  depends on the Wasserstein distance between the sequence of discretely supported marginals and the true marginal distributions.

 In related literature, one of the earliest works in model-independent option pricing can be found in \cite{hobson1998robust}, where the author focuses on obtaining model-independent bounds on the price of a lookback option by formulating it as a Skorokhod embedding problem. \cite{brown2001robust}, \cite{cox2011robust},\cite{beiglbock2020geometry}, and the references therein provide detailed insights into the applications of Shorokhod embedding techniques for robust pricing and hedging of derivatives. 
 
In order to ensure that the discrete and the true underlying marginals are consistent with the available call prices with maturities $0<t_1<T$. The discretization approach in \cite{schmithals2019contributions} provides marginals consistent with the call prices. To obtain an idea about how close the MOT problem with the discretely supported marginals would be to the value of the MOT problem for the true underlying measure, we use the convergence results from \cite{schmithals2019contributions}. While the problem (\ref{generalized_objective_at_t1}) is not a standard MOT problem, we solve the sequence of linear problems that one obtains for the discretely supported marginals, similar to the problem (\ref{generalized_objective}), and compare the hedging performance with $(\ref{generalized_objective})$.

To the best of our knowledge, the standard approach in literature is to view the hedging problem corresponding to the target option with payoff $c(X,Y)$ as the dual problem to the primal problem MOT problem (\ref{OT})-(\ref{MOT}), starting from the pioneering approach in \cite{beiglbock2013model}. The dual approach to the problem (\ref{OT})-(\ref{MOT}) is a more generalized version of a semi-static hedging problem over a set of functions that can be thought of as call/put options with maturities $0<t_1<T$ and dynamic positions in the shares rebalanced at $t_1$ and $T$. This is explained in greater detail in Section \ref{Sec:Convergence results}.

For a higher dimensional stock price process where the resulting MOT problem is known as a \textit{multi-marginal martingale optimal transport} (MMOT)  problem, and even for the case when the pay-off depends on the asset price at more than two time-points, solving the dual problem to the primal pricing problem provides a significant reduction in the computational cost. The MMOT problem was introduced in \cite{lim2016multi} and the reader can also refer to \cite{henry2016maximum}, \cite{cox2019root}, \cite{beiglbock2020geometry}, \cite{nutz2020multiperiod}, \cite{eckstein2021robust} for studies on multi-marginal problems. 

Given that the investor already knows which call/put options they want to include in their portfolio $h(X,w)$, \textbf{our hedging problem (\ref{generalized_objective}) (respectively (\ref{generalized_objective_at_t1})) allows the investor to compute the weights of the components of the hedging portfolio (restricted to cash, stock and options maturing at $t_1$) that minimizes the worst possible expected error at time $T$ (respectively $t_1$)}. Under any unforeseen situation thrown by nature, the hedging error of this portfolio should be bounded above by the value of (\ref{generalized_objective}). The approach of viewing the hedging errors as (\ref{generalized_objective})-(\ref{generalized_objective_at_t1}) for two time-points $0<t_1<T$ also reduces to a \textit{min-max} problem for the case when the underlying marginal distributions are discretely supported, which could be of independent interest to the reader.

The main contributions of our paper are :
\begin{enumerate}
\item Extend the static hedging problem in \cite{carr2014static} to a model-free framework to construct a static hedging portfolio of cash, stock, and short-maturity \textit{European} call options to hedge a longer-maturity target \textit{European} call or a path-dependent option.
     \item Obtain the worst possible bounds on the hedging error while minimizing with respect to the portfolio weights by formulating the maximization problem as a modified version of the MOT problem.
     \item Prove the theoretical convergence of the min-max problem for the discrete marginals to the inf-sup problem (\ref{generalized_objective}) in the continuous case. 
     \item Formulate the corresponding max-min optimization problem for (\ref{generalized_objective_at_t1}), explain the utility of this approach, and compare it with the results for (\ref{generalized_objective_at_t1}). 
     \item Compare the worst-case hedging error at short-maturity $t_1$ for the hedging portfolio (having options with maturities $t_1$ and $T$) obtained using the standard dual problem of the primal problem MOT problem (\ref{OT})-(\ref{MOT}) with the hedging portfolio obtained using our approach in (\ref{generalized_objective_at_t1}). However, our hedging portfolio in this case has options with maturities $t_1$ as well as $T$, similar to the dual problem.  
\end{enumerate}

The outline of the paper is as follows: Section \ref{section:Framework} introduces the notations and important results from \cite{schmithals2019contributions} and \cite{bauerle2021consistent}, giving the background for the MOT problem. We discuss the associated numerical schemes from \cite{schmithals2019contributions} and \cite{bauerle2021consistent} to obtain the discretely supported marginal distributions of the stock price process and their results on the corresponding Wasserstein distance between the discrete marginal distributions and the true underlying marginals. Our problem at hand is explained in greater detail in Section \ref{Section:Hedging problem}, and the alternative \textit{max-min} formulation and its financial interpretation are given. The dual problem and associated convergence results for our problem (\ref{generalized_objective}) are provided in Section \ref{Sec:Convergence results}. In Section \ref{Sec:Numerical examples}, numerical examples are provided to test the efficiency of the numerical scheme and the associated upper bounds for both the pricing of options and the hedging problem at hand. Section \ref{Sec:Conclusion} gives the conclusion and discussions on possible future work.

\section{Framework and Preliminaries}\label{section:Framework}
We begin with a financial market with one risk-free asset, referring to the cash deposited in a bank account, and one risky asset, $S_t$, denoting the stock price path. Let $0=t_0<t_1<...<t_N=T$ denote the time points at which the stock price process needs to be evaluated for obtaining the marginal distributions, with $T$ being the final time point. We follow the terminologies and results in \cite{schmithals2019contributions} and \cite{bauerle2021consistent} to ensure consistency. 

It is assumed that the risk-free asset pays no interest, $r=0$, and the risky asset with price process $(S_t)$ is denoted by $(S_{t_1}, S_{t_2},...S_{t_N})$ with initial value $S_0 = 1$. Throughout this paper, for simplicity, we consider only two time points, i.e., $N=2$.

Following standard conventions, the random variables $(S_{t_1}, S_T)$, denoted henceforth by $(X,Y)$, take only non-negative values and are defined on a probability space $(\Omega,\mathcal{F},\mathbb{P})$. \textbf{There are no underlying model assumptions on the stock price process, but the market is always assumed to be free of arbitrage.} This guarantees the existence of a risk-neutral measure for the underlying stock price process by the First Fundamental Theorem of Asset Pricing (\cite{shreve2004stochastic}).

Restricting to the case of $N=2$ time-points,  let $C_{t_1}(k), C_T(k)$ denote the prices of \textit{European} call options (written on the underlying stock price process) at the initial time $t_0$, with strike price $k\in \mathbb{R}_{+}$ and maturities $t_1,t_2=T$ respectively.

If  $P(\mathbb{R}_{+})$ denotes the set of all probability measures on $\mathbb{R}_{+}$, then the First Fundamental Theorem of Asset Pricing ensures the existence of a measure $\mu\in P(\mathbb{R}_{+})$ which satisfies
\begin{align}
    C(k)=\int (x-k)^{+}\mu(dx),\hspace{0.5mm} k\geq 0
\end{align}
and the associated risk-neutral distribution by \cite{breeden1978prices} is
\begin{align}
\label{BL hint}
    \mu((-\infty,x])=1+C'(x+),x\in\mathbb{R}.
\end{align}

\subsection{Convex order and associated properties}
 
We first introduce what one means by the convex ordering of two measures and state the associated results. Lemmas (\ref{convex1}) and (\ref{convex2}) highlight the close relationship between the convex ordering of the underlying marginal distributions and the martingale property (\ref{MOT}), and also with the pricing functions under consideration.
\begin{defn}
    Two measures $\mu,\nu$ on $\mathbb{R}$ are said to be in convex order, denoted by $\mu\leq_{c}\nu$, if for any convex function $f:\mathbb{R}\rightarrow\mathbb{R}$ such that the integrals exist
    \begin{align}
        \int_{\mathbb{R}}f(x)\mu(dx)\leq \int_{\mathbb{R}}f(x)\nu(dx).
    \end{align}
\end{defn}

The following result by \cite{strassen1965existence} relates the convex ordering of measures with the martingale property of the associated random variables.
\begin{lem}
\label{convex1}
    Suppose $\mu,\nu \in P(\mathbb{R}_{+}). $ Then $\mu \leq_{c}\nu$ is equivalent to the existence of a probability space $(\Omega,\mathcal{F},\mathbb{P})$ and non-negative random variables $X,Y$ on it such that $X$ has distribution $\mu$ and $Y$ has distribution $\nu$ and $X=\mathbb{E}[Y|X]$.
\end{lem}

Lemma \ref{convex2} provides a relationship between the convex ordering and the associated values of the call price functions. 
\begin{lem}
\label{convex2}
    Let $\mu,\nu\in P(\mathbb{R}_{+})$ and denote by $C_{\mu}$ and $C_{\nu}$ the respective consistent pricing functions. Suppose that $\int x\mu(dx) =\int x \nu(dx)=1.$ Then $\mu\leq_{c}\nu$ is equivalent to $C_{\mu}\leq C_{\nu}.$
\end{lem}
Next, we define what we mean by consistent pricing measures and the discretization schemes for obtaining such measures from the observed \textit{European} call option prices.

\subsection{Marginals with bounded support}
As mentioned earlier, in practice,  corresponding to a fixed maturity $t_i$, one observes only finitely many call prices $c^{i}_{0}>...>c^{i}_{n_i}>0$ associated with the strikes $0\leq k^{i}_{0}<..<k^{i}_{n_i},n_i \in \mathbb{N},i=1,2$. Hence, one cannot apply (\ref{BL hint}) to obtain the true risk-neutral marginal distributions at $t_i,i=1,2$. This leads us to the following definition.

\begin{defn}\cite{bauerle2021consistent}
    Let for $i=1,2$,
    \begin{align}
       P_i := \{ \mu\in P(\mathbb{R}_{+}):c^{i}_{j}=\int (x-k^{i}_{j})^+\mu(dx),j=0,..,n_i,\int x \mu(dx)=S_0\}  
    \end{align}
   be the set of all pricing measures that are consistent with the observable call prices having maturity $t_i$.
\end{defn}
We recall the earlier observation that the given MOT problem (\ref{OT})-(\ref{MOT}) reduces to an LP problem for the case of discretely supported marginals of the stock price process, $\mu$ and $\nu$ at time $t_1$ and $T$ respectively. However, given that the true underlying marginal distributions of the stock price process need not be discretely supported, one would like to construct a sequence of discrete marginals that can be shown to converge to the true underlying marginals $\mu\in P_1$ and $\nu\in P_2$ under some given metric, along with the convergence of the corresponding MOT problems.

To achieve this, in \cite{bauerle2021consistent}, the authors first make the following assumptions :
\begin{enumerate}
    \item There is a strike price $K>0$, with call prices equal to zero for every strike greater than or equal to $K$.
    \item A finite number of call prices $c^{i}_0>...>c^{i}_{n_i}=0$ are available for strikes $0=k^{i}_{0}<...<k^{i}_{n_i}=K$, with $c^{i}_{0} =S_0$.
\end{enumerate}  

Let $K_i:=\{k_{0}^{i}<...<k_{n_i}^{i}\},i=1,2$ denote the set of strike prices for which the call prices are observable. To obtain the call prices at all strikes,  in \cite{bauerle2021consistent}, the authors choose the functions $C^{*}_{\mu}, C^{*}_{\nu}$ that result from linearly interpolating the call prices available at each time point, $t_i,i=1,2$. The resulting function, with $C^*_{\mu}(k^{1}_{j})=c^{1}_{j},j=0,...,n_{1}$ and $k \in [k^{1}_j,k^{1}_{j+1}),j=0,..,n_{1},$ is given by
\begin{align}
    \label{interpolation}
    C^{*}_{\mu}(k)=\frac{k^{1}_{j+1}-k}{k^{1}_{j+1}-k^{1}_{j}}C^{*}_{\mu}(k^{1}_j)+\frac{k-k^{1}_{j}}{k^{1}_{j+1}-k^{1}_{j}}C^{*}_{\mu}(k^{1}_{j+1}),
\end{align}
and one similarly obtains $C^{*}_{\nu}$.

The special discrete marginals, consistent with the call price functions, are then constructed in \cite{bauerle2021consistent} using (\ref{interpolation}) as follows :

\begin{lem}\cite{bauerle2021consistent}
    \label{nicole_marginals}
    The measure $\mu^*$ consistent with $C^{*}_{\mu}$ is a discrete measure of the form
    \begin{align}
        \label{nicole_marginal_dist}
        \mu^{*} =\sum_{j=0}^{n}\left[\frac{C^{*}_{\mu}(k_{j+1})-C^{*}_{\mu}(k_j)}{k_{j+1}-k_{j}}-\frac{C^{*}_{\mu}(k_{j})-C^{*}_{\mu}(k_{j-1})}{k_{j}-k_{j-1}}\right]\delta_{k_j},
    \end{align}
    where we set
    \begin{align}
        \begin{split}
            \frac{C^{*}_{\mu}(k_{n+1})-C^{*}_{\mu}(k_n)}{k_{n+1}-k_{n}}=0\hspace{0.2cm}\text{and}\\
            \frac{C^{*}_{\mu}(k_{0})-C^{*}_{\mu}(k_{-1})}{k_{0}-k_{{-1}}}=-1
        \end{split}
    \end{align}
    and $\delta_x$ is the Dirac measure on point $x$.
\end{lem}
Lemma \ref{nicole_marginals} allows one to obtain the marginal distribution from the available call option prices at finitely many strikes. It does not require any restrictions on the spacing between the strike points, and hence can be readily applied.

The following lemma shows that $\mu^{*}$ is the maximal element of the set $P_1$ with respect to convex ordering.
\begin{lem}\cite{bauerle2021consistent}
    \label{maximal}
    Suppose that $\mu\in P_1$, i.e., $\mu$ is another probability measure consistent with the observable call prices in $P_1$. Then
    \begin{align}
        \mu \leq_{c} \mu^{*}.
    \end{align}
\end{lem}

The Wasserstein distance between two probability measures is defined as : 

Let
\begin{align}
\label{transport_2_timepoints}
    \mathcal{P}(\mu,\mu^{*})=\{\mathbb{P}\in P(\mathbb{R}^2):\mu(B_1)=\mathbb{P}(B_1\times\mathbb{R}),\mu^{*}(B_2)=\mathbb{P}(\mathbb{R}\times B_2),B_1,B_2\in\mathcal{B}(\mathbb{R})\}.
\end{align}

\begin{defn}
\label{wassertein_def}
    The \textit{Wasserstein distance} of two probability measures $\mu,\mu^{*}\in P(\mathbb{R})$ is given by
    \begin{align}
        \label{wasserstein1}
        W(\mu,\mu^{*}) = \inf_{\mathbb{P}\in\mathcal{P}(\mu,\mu^{*})}\int |x-y|\mathbb{P}(dx,dy).
    \end{align}
\end{defn}
It can be observed from Definition \ref{wasserstein1} that the Wasserstein distance of two measures $\mu,\mu^*\in\mathcal{P}(\mathbb{R})$ is a special case of the usual optimal transport problem (\ref{OT})-(\ref{marginal}) with the cost function given by $c(x,y)=|x-y|$. \cite{bauerle2021consistent} lists the equivalent representations of the Wasserstein distance as given below. 
\begin{rem}\cite{bauerle2021consistent}
    \begin{enumerate}
        \item If $F_{\mu}$ and $F_{\mu^{*}}$ are the cumulative distribution functions of $\mu$ and $\mu^{*}$, the following equality also holds \cite{dall1956sugli}
          \begin{align}
          \label{wasserstein2}
              W(\mu,\mu^{*})=\int_{-\infty}^{\infty}|F_{\mu}(t)-F_{\mu^{*}}(t)|dt.
          \end{align}
          \item  A dual representation of the Wasserstein distance is as follows:
          \begin{align}
              \label{wasserstein3}
              W(\mu,\mu^{*})=\sup_{f\in C_1(\mathbb{R})}\int f(x)(\mu-\mu^{*})(dx),
          \end{align}
          where, $C_1(\mathbb{R}):=\{f:\mathbb{R}\rightarrow\mathbb{R}:f\hspace{0.25cm}\text{is Lipschitz-continuous with constant $1$}\},$ \cite{villani2009optimal}.
    \end{enumerate}
\end{rem}

.
\begin{thm}\cite{bauerle2021consistent}
\label{thm_bound_nicole}
    Let $\mu\in P_1$ with supp($\mu$) $\subset[0,K]$. Moreover choose $k_j= \frac{jK}{2^n}, j=0,...,2^n,n\in\mathbb{N}$. Then we have
    \begin{align}
        W(\mu,\mu^*)=2\cdot\sum_{j=0}^{2^n-1}\sup_{k\in[k_j,k_{j+1})}|C_{\mu^*}(k)-C_{\mu}(k)|\leq\frac{K}{2^n}.
    \end{align}
    If we additionally assume that $C_{\mu}\in C^2(\mathbb{R}_{+})$, then for any $n\in\mathbb{N}$, we have
    \begin{align}
        W(\mu,\mu^*)\leq \frac{T_{\mu}\cdot K^2}{2^{n+1}},
    \end{align}
    where $T_{\mu}=\sup_{\kappa\in[0,K]}|C_{\mu}^{''}(\kappa)|.$
\end{thm}
Under the given assumption that the true underlying distributions $\mu,\nu$ have a compact support $[0,K]$ and $\mu\leq_{c}\nu$, Theorem \ref{thm_bound_nicole} shows that the dominating measure $\mu_n^*$ converges to the true marginal as the available call prices become dense over uniformly specified set of strikes. However, for the case when the true underlying marginal distributions have unbounded support, e.g., the \textit{Black-Scholes} model, constructing the discrete marginals $\mu^*$ and $\nu^*$ from the observable call prices using equation (\ref{marginal}) may not yield any feasible solution to the LP for the strict MOT problem. Further, one cannot directly extend the convergence results of Theorem 5.1 of \cite{bauerle2021consistent} to the case when the true marginal distributions have unbounded support. This leads us to the discretization scheme and corresponding convergence results from \cite{schmithals2019contributions}, which we use for our problems (\ref{generalized_objective}) and (\ref{generalized_objective_at_t1}).

\begin{rem}
    The question naturally arises whether the solution to the discrete MOT problem converges to the solution under the true measure as call prices become available over an increasingly dense set of strikes. An affirmative answer to this (under certain restrictions) can be found in the convergence results in \cite{bauerle2021consistent}.
\end{rem}

\subsection{Marginals with unbounded support}\label{subsec:marginals_unbdd}
We assume here that the theoretical marginals $\mu,\nu\in\mathcal{P}(\mathbb{R}_+)$ have unbounded support. We use the notation from \cite{schmithals2019contributions} and denote the approximating measures in this case by $\mu_{n}^{\infty},\nu_{n}^{\infty}$, for a given number of discretization points, $n$. Let 
$Z_{n}^{\infty}:=\{ k_{j}^{n}|j=0,\dots,n\}$ denote the set of strike prices for every $n\in\mathbb{N}$ with $k_0=0$. Differing from \cite{schmithals2019contributions}, \textbf{we do not assume the strike prices to be evenly spaced}. The associated option prices are given by
\begin{align*}
    C_{n}^{\mu,\infty}:=\{C_{\mu}(k)\mid k\in Z_{n}^{\infty}\}\hspace{1cm}\text{and}\hspace{1cm}C_{n}^{\nu,\infty}:=\{C_{\nu}(k)\mid k\in Z_{n}^{\infty}\}.
\end{align*}
Following \cite{schmithals2019contributions} we define the candidate functions $C_{\mu_{n}^{\infty}},C_{\nu_{n}^{\infty}}\in\mathcal{K}^C$ consistent with the prices in $C_{n}^{\mu,\infty}$ and $C_{n}^{\nu,\infty}$ such that $\mu_{n}^{\infty}\leq_{c}\nu_{n}^{\infty}$. For $k\in[k_{j}^n,k_{j+1}^n),j=0,\dots,n-1$, define
\begin{align}
    \label{call_price_inter_unbdd}
    &C_{\mu_{n}^{\infty}}(k):=\frac{k_{j+1}^n-k}{k_{j+1}^n-k_{j}^n}C_{\mu}(k_j^n)+\frac{k-k_j^n}{k_{j+1}^n-k_j^n}C_{\mu}(k_{j+1}^n)\\
     &C_{\nu_{n}^{\infty}}(k):=\frac{k_{j+1}^n-k}{k_{j+1}^n-k_{j}^n}C_{\nu}(k_j^n)+\frac{k-k_j^n}{k_{j+1}^n-k_j^n}C_{\nu}(k_{j+1}^n).
\end{align}
Let $k_{\mu,0}^n$ and $k_{\nu,0}^n$ denote the smallest zeros of the continuations of $C_{\mu_{n}^{\infty}}$ and $C_{\nu_{n}^{\infty}}$ on $[k_{n-1}^nk_n^n)$ to $(k_n^n,\infty)$ defined by
\begin{align}
\label{zeros_of_interpolated_lines}
    &k_{\mu,0}^n=\inf\bigg\{k\in(k_n^n,\infty)\bigg\vert\frac{k_{n}^n-k}{k_{n}^n-k_{n-1}^n}C_{\mu}(k_{n-1}^n)+\frac{k-k_{n-1}^n}{k_{n}^n-k_{n-1}^n}C_{\mu}(k_{n}^n)=0 \bigg\},\\
     &k_{\nu,0}^n=\inf\bigg\{k\in(k_n^n,\infty)\bigg\vert\frac{k_{n}^n-k}{k_{n}^n-k_{n-1}^n}C_{\nu}(k_{n-1}^n)+\frac{k-k_{n-1}^n}{k_{n}^n-k_{n-1}^n}C_{\nu}(k_{n}^n)=0 \bigg\}.
\end{align}
Define
\begin{align}
    \label{last_call_price_mu}
    C_{\mu_{n}^{\infty}}(k):=\begin{cases}
        \frac{k_{\mu,0}^n-k}{k_{\mu,0}^n-k_n^n}C_{\mu}(k_n^n),& \text{if $k\in(k_n^n,k_{\mu,0}^n),$}\\
        0,&\text{if $k\in[k_{\mu,0}^n,\infty)$}.
    \end{cases}\
\end{align}
For $C_{\nu_n^{\infty}}$, one needs to distinguish between two cases depending on the values of $k_{\mu,0}^n$ and $k_{\nu,0}^{n}$ as follows :

\begin{itemize}
    \item Case $1$ : If $k_{\mu,0}^n\leq k_{\nu,0}^n$ define
    \begin{align}
    \label{call_mu}
        C_{{\nu}_n^\infty}(k):=\begin{cases}
        \frac{k_{\nu,0}^n - k}{k_{\nu,0}^n - k_{n}^n}C_{\nu}(k_n^n),  & \text{if} ~~k\in(k_n^n,k_{\nu,0}^n),  \\
       0, & \text{if}~~ k\in [k_{\nu,0}^n,\infty).
        \end{cases}
    \end{align}
    \item Case $2$ : If $k_{\mu,0}^n>k_{\nu,0}^n$ define
    \begin{align}
    \label{call_nu}
        C_{{\nu}_n^\infty}(k):=\begin{cases}
        \frac{k_{\mu,0}^n - k}{k_{\mu,0}^n - k_{n}^n}C_{\nu}(k_n^n),  & \text{if} ~~k\in(k_n^n,k_{\mu,0}^n),  \\
       0, & \text{if}~~ k\in [k_{\mu,0}^n,\infty).
        \end{cases}
    \end{align}
\end{itemize}
Therefore, equations (\ref{call_price_inter_unbdd})-(\ref{call_nu}) define the call option price functions $C_{\mu_n^\infty},C_{\nu_n^\infty}\in\mathcal{K}^C$. From these functions, one can derive probability measures $\mu_n^\infty,\nu_n^\infty\in\mathcal{P}(\mathbb{R}_+)$  with expected value equal to one  using equation (\ref{nicole_marginal_dist}). The resulting measures are in convex order by construction, i.e. $\mu_n^\infty\leq_c\nu_n^\infty$, and they are discrete measures of the form
\begin{align}
    \mu_n^\infty:=\sum_{j=0}^{n-1}w_j^n\delta_{k_j^n}+\mu_{n,r}^\infty,
\end{align}
where $w_j^n=\frac{C_\mu(k_{j+1}^n)-C_\mu(k_j^n)}{k_{j+1}^n-k_j^n}-\frac{C_\mu(k_{j}^n)-C_\mu(k_{j-1}^n)}{k_{j}^n-k_{j-1}^n}$,
and 
\begin{align}
    \nu_n^\infty :=\sum_{n=0}^{n-1}v_i^n\delta_{k_i^n}+v_{n,r}^\infty.
\end{align}
$v_i^n=\frac{C_v(k_{i+1}^n)-C_v(k_i^n)}{k_{i+1}^n-k_i^n}-\frac{C_\nu(k_i^n)-C_\nu(k_{i-1}^n)}{k_i^n -k_{i-1}^n}$ and
\begin{align}
\begin{split}
    v_{n,r}^{\infty}&:=(v_n^n\delta_{k_n^n}+v_{\mu}^n\delta_{k_{\mu,0}^n})\mathbbm{1}_{\{k_{\mu,0}^n>k_{\nu,0}^n\}}+v_{\nu}^n\delta_{k_{\nu,0}^n}\mathbbm{1}_{\{k_{\mu,0}^n\leq k_{\nu,0}^n\}}\\
    &:=\bigg(\bigg[\frac{-C_\nu(k_n^n)}{k_{\mu,0}^n-k_n^n}-\frac{C_\nu(k_n^n)-C_\nu(k_{n-1}^n)}{k_n^n-k_{n-1}^n}\bigg]\delta_{k_n^n}+\frac{C_\nu(k_n^n)}{k_{\mu,0}^n-k_n^n}\delta_{k_{\mu,0}^n}\bigg)\mathbbm{1}_{\{k_{\mu,0}^n> k_{\nu,0}^n\}}\\
    &+\frac{C_\nu (k_{n-1}^n)-C_\nu(k_n^n)}{k_n^n-k_{n-1}^n}\delta_{k_{\nu,0}^n}\mathbbm{1}_{\{k_{\mu,0}^n\leq k_{\nu,0}^n\}}.
\end{split}
\end{align}

   Defining $ h^\nu(k_{4^n}^n) :=\mathbbm{1}_{\{F_\nu (k_{4^n}^n)\geq F_{\nu_{n}^{\infty}}(k_{4^n}^n)\}}$ in \cite{schmithals2019contributions} the authors prove that
   \begin{align}
   \label{wasserstein_bound_mu}
       W(\mu,\mu_n^{\infty})\leq \frac{1+C_{\mu}^,(k_{4^n}^n)}{2^n} + 2 C_{\mu}(k_{\mu,0}^n)=\frac{F_{\mu}(k_{4^n}^n)}{2^n}+2C_{\mu}(k_{\mu,0}^n),
   \end{align}
    and
    \begin{align}
    \begin{split}
    \label{wasserstein_bound_nu}
        W(\nu,\nu_n^{\infty})&\leq \frac{F_{\nu}(k_{4^n}^n)}{2^n}+ C_\nu\bigg(\max\{k_{\nu,0}^n,k_{\mu,0}^n\}\bigg)\\
        &+C_\nu \bigg(\max\{k_{\nu,0}^n,k_{\mu,0}^n\}\bigg)\bigg(\mathbbm{1}_{\{k_{\mu,0}^n\leq k_{\nu,0}^n\}}+\mathbbm{1}_{\{k_{\mu,0}^n>k_{\nu,0}^n\}}h^\nu(k_{4^n}^n)\bigg)\\
       & +(k_{\mu,0}^n-k_{4^n}^n)\bigg(C_{\nu}^{'}(k_{\mu,0}^n)-C_{\nu}^{'}(k_{4^n}^n)\bigg)\mathbbm{1}_{\{k_{\mu,0}^n>k_{\nu,0}^n\}}(1-h^\nu(k_{4^n}^n)).
    \end{split}
        \end{align}

We will now state our static model-free hedging problem and derive some results by observing structures similar to those that appear in the pricing problems described in this section.
\section{Problem at hand}\label{Section:Hedging problem}

We begin by defining $C(S(t),t, K, T)$ to be the price of a \textit{European} call option at time $t\in[0, T]$, with maturity $T$, strike price $K$, and the underlying stock price $S(t)$. Let $c_{\text{target}}(S(t_1), S(T))$ be the payoff of a path-dependent option at expiry $T$ with $0<t_1<T$.

A writer of this path-dependent option can hedge their short position with actively traded stocks, options, or any other instruments at their disposal. A hedging portfolio comprising of cash, the underlying stock, and \textit{European} call options on the same underlying asset, with maturity $t_1(<T)$, at time $0$ is then given by
\begin{align}
    \label{hedge0}
   w_0 + w_1 S(0)+\sum_{i=2}^{M}w_i C(S(0),0,K_{i-1},t_1),
\end{align}
where $w_i$'s are the weights and $K_i$'s are the strikes of the short maturity options.

The absolute expected worst-case error of this hedging portfolio at time $T$ is defined by
\begin{align}
    \label{MOT_hedge}
    \begin{split}
       P^{H}_{T}(\mu,\nu)&:=\inf_{w_i,0\leq i\leq M}\sup_{\mathbb{P}\in\mathcal{M}(\mu,\nu)}\mathbb{E}_{\mathbb{P}}\bigg[\bigg|[c_{\text{target}}(S(t_1),S(T))\\
       &-w_0-w_1S(T)-\sum_{i=2}^{M}w_i C(S(t_1),t_1,K_{i-1},t_1)\bigg|\bigg],
    \end{split}
   \end{align}
where $H$ is used to denote the hedging problem.  

As explained in Section \ref{sec:Introduction}, from the option writer's perspective, it is more important to compute the hedging error when the options in their hedging portfolio expire.  We will refer to expressions such as the one on the right in (\ref{MOT_hedge}) as absolute hedging error. We focus our numerical analysis on this problem in Section \ref{Sec:Numerical examples}. The inside maximization problem in (\ref{MOT_hedge}) is an MOT problem as described in Section \ref{sec:Introduction}. So we can talk here about the modified MOT (Mod-MOT) problem.

The Mod-MOT problem for the absolute hedging error at time $t_1$ (the maturity of the short-term options in the hedging portfolio) is given by
\begin{align}
    \label{MoMOT_hedge_t1}
    \begin{split}
       P^{H}_{t_1}(\mu,\nu)&:=\inf_{w_i,0\leq i\leq M}\sup_{\mathbb{P}\in\mathcal{M}(\mu,\nu)}\mathbb{E}_{\mu}\bigg[\bigg|\mathbb{E}_{\mathbb{P}^x}[c_{\text{target}}(S(t_1),S(T))|S(t_1)]\\
       &-w_0-w_1S(t_1)-\sum_{i=2}^{M}w_i C(S(t_1),t_1,K_{i-1},t_1)\bigg|\bigg], 
    \end{split}
   \end{align}
where $\mathbb{P}^x$ denotes the disintegration of the joint probability distribution with respect to $\mu$.

We assume that the only information available to the writer is the observable call option prices with maturities $t_1$ and $T$, respectively, at a finite number of strikes. As explained in Section \ref{sec:Introduction}, it is more important for the investor to look at the problem (\ref{MoMOT_hedge_t1}) when the options in their hedging portfolio expire than the problem (\ref{MOT_hedge}). Therefore, it is necessary to formulate the problem (\ref{MoMOT_hedge_t1}) for the case of discretely supported marginals obtained using these call option prices.  Let $\alpha$ with $\alpha(x)=\sum_{i=1}^{N_1}\delta_{x_{i}}(x)\alpha_{i}$ and $\beta$ with $\beta(y)=\sum_{j=1}^{N_2}\delta_{y_{j}}(y)\beta_{j}$ be the finitely supported marginal distributions at times $t_1$ and $T$ obtained using the available option prices as described in Section \ref{subsec:marginals_unbdd}. Then the corresponding min-max problem for the absolute hedging error at time $t_1$, with cash, stock, and two options (for simplicity) in the hedging portfolio, is
\begin{align}
    \label{ourLP}
    \begin{split}
        &\min_{w_i,0\leq i\leq3}\max_{p\in\mathbb{R}^{N_{1}\times N_{2}}}\sum_{i=1}^{N_1}\alpha_{i}\bigg|\sum_{j=1}^{N_2}p_{j|i}c_{\text{target}}(x_i,y_j)-w_0-w_1x_i- w_2(x_i -K_1)^{+}-w_3(x_i-K_2)^{+}\bigg|\hspace{0.2cm}\\
        &\text{subject to}
        \sum_{j=1}^{N_{2}}p_{i,j}=\alpha_i,\hspace{0.2cm}\text{for}\hspace{0.2cm} i=1,..,N_1,\\
    &\sum_{i=1}^{N_{1}}p_{i,j}=\beta_j,\hspace{0.2cm}\text{for}\hspace{0.2cm} j=1,..,N_2,\\
    &\sum_{j=1}^{N_2}p_{i,j}(y_j-x_i)=0,\hspace{0.2cm}\text{for}\hspace{0.2cm} i=1,..,N_1,\\
    & w_0,w_1,w_2,w_3\in\mathbb{R},\\
   & p_{i,j}\geq0,
    \end{split}
\end{align}
where  $p_{j|i}=\frac{p_{i,j}}{\alpha_i}$  denotes the conditional probability that $Y=y_j$, given $X=x_i$. 

Using the fact that $\alpha_i> 0, i=1\dots N_1$, the objective function in (\ref{ourLP}) can be simplified to
\begin{align}
    \label{our_objective}
    \min_{w_i,0\leq i \leq3}\max_{p\in\mathbb{R}^{N_{1}\times N_{2}}}\sum_{i=1}^{N_1}\bigg|\sum_{j=1}^{N_2}p_{i,j}c_{\text{target}}(x_i,y_j)-\alpha_i\{w_0+w_1x_i+w_2(x_i -K_1)^{+}+w_3(x_i-K_2)^{+}\}\bigg|.
\end{align}

\begin{rem}
\begin{enumerate}
   \item Put options can be included in the hedging portfolio in (\ref{our_objective}) by using put-call parity.
    \item The resulting hedge is valid only till time $t_1$, when the short-maturity options expire. At maturity $t_1$, the writer can either close their position on the target option with maturity $T$ or set up an entirely new hedge with available call/put options.
     \item  If the true dynamics of the underlying stock price process are known and satisfy one-factor Markovian conditions, our hedging problem is similar to the static-hedging approaches in \cite{carr2014static} and \cite{banerjeemultiperiod} where the exact weights of the hedging portfolio (constituting only the options with shorter maturities than the target maturity $T$) are computed using Gauss-Hermite, Gauss-Laguerre, and Gaussian quadrature algorithms. 
    \item Corresponding to different choices for weights $w_i,0\leq i \leq M$, in the hedging portfolio, one obtains various solutions to the maximization problem, and the associated joint-probability distributions $\{p_{i,j}\}_{1\leq i \leq M_1,1\leq j\leq M_2}$. Hence, formulating the hedging problem as a min-max problem allows one to capture this effect. 
     \item  If the investor knows the underlying joint probability distribution $\mathbb{P}\in\mathcal{M}(\mu,\nu)$, it is possible to compute the weights that minimize the corresponding hedging error under the chosen distribution. One can then obtain the maximum loss for different combinations of probability distributions, with the objective function of interest given by
    \begin{align}
        \label{MoMOT_hedge2}
        \begin{split}
           P_{t_1}^{H_2}(\mu,\nu)&:=\sup_{\mathbb{P}\in\mathcal{M}(\mu,\nu)}\inf_{w_i}\mathbb{E}_{\mu}\bigg[\bigg|\mathbb{E}_{\mathbb{P}^x}
           [c_{\text{target}}(S(t_1),S(T))|S(t_1)]\\
           &-w_0-w_1S(t_1)-\sum_{i=2}^{M}w_i C(S(t_1),t_1,K_{i-1},t_1)\bigg|\bigg]. 
        \end{split}
      \end{align}
      Under the assumption that the marginal distributions $(\mu,\nu)$ of the stock price process are finitely supported, the problem (\ref{MoMOT_hedge2}) reduces to a \textit{max-min} problem.
    \item The two problems (\ref{MoMOT_hedge_t1}) and (\ref{MoMOT_hedge2}) are not necessarily equal and are related by the following inequality
    \begin{align}
    \label{relaxedMOT_inequality}
    \begin{split}
       &\sup_{\mathbb{P}\in\mathcal{M}(\mu,\nu)}\inf_{w_l,0\leq l \leq M}\mathbb{E}_{\mu}\bigg[\bigg|\mathbb{E}_{\mathbb{P}^x}[c_{\text{target}}(S(t_1),S(T))|S(t_1)]\\
       &-w_0-w_1S(t_1)-\sum_{i=2}^{M}w_i C(S(t_1),t_1,K_{i-1},t_1)\bigg|\bigg]\\
        &\leq \inf_{w_l,0\leq l \leq M}\sup_{\mathbb{P}\in\mathcal{M}(\mu,\nu)}\mathbb{E}_{\mu}\bigg[\bigg|\mathbb{E}_{\mathbb{P}^x}[c_{\text{target}}(S(t_1),S(T))|S(t_1)]\\
        &-w_0-w_1S(t_1)-\sum_{i=2}^{M}w_i C(S(t_1),t_1,K_{i-1},t_1)\bigg|\bigg].
    \end{split}
    \end{align}
\end{enumerate}
\end{rem}
The reader can refer to well-established results like Sion's min-max Theorem in  \cite{sion1958general}, \cite{terkelsen1972some} for sufficient conditions on the function $ \mathbb{E}_{\mu}\bigg[\bigg|\mathbb{E}_{\mathbb{P}^x}[c_{\text{target}}(S(t_1),S(T))|S(t_1)]-w_0-w_1S(t_1)-\sum_{l=2}^{M}w_l C(S(t_1),t_1,K_{i-1},t_1)\bigg|\bigg]$ for equality to hold in (\ref{relaxedMOT_inequality}).

If we look at the corresponding max-min problem for (\ref{MoMOT_hedge_t1}) for the discretely supported marginals, we get
\begin{align*}
&\max_{p_{i,j}}\min_{w_l,0\leq l \le 3}\sum_{i=1}^{N_1}\alpha_i|\sum_{j=1}^{N_2}\frac{p_{i,j}c(x_i,y_j)}{\alpha_i}\\
&-
    w_0-w_1x_i-w_2(x_i-K_1)^{+}-w_3(x_i-K_2)^+|\hspace{0.2cm}\\
    &\text{subject to}
    \sum_{j=1}^{N_2}p_{i,j}=\alpha_i, i=1,..,N_1,\\
    &\sum_{i=1}^{N_1}p_{i,j}=\beta_j, j=1,..,N_2,\\
    &\sum_{j=1}^{N_2}p_{i,j}(y_j-x_i)=0, i=1,..,N_1,\\
   & w_0,w_1,w_2,w_3\in\mathbb{R},\\
   & p_{i,j}\geq0,
\end{align*}
where the inner minimization problem is
\begin{align}
\label{inner_min_prob}
\begin{split}
   &\min_{w_l,0\leq l\leq3}\sum_{i=1}^{N_1}\alpha_i|\sum_{j=1}^{N_2}\frac{p_{i,j}c(x_i,y_j)}{\alpha_i}\\
    &-
    w_0-w_1x_i-w_2(x_i-K_1)^{+}-w_3(x_i-K_2)^+|,\\
    &w_0,w_1,w_2,w_3\in\mathbb{R}.  
\end{split}
    \end{align}

  For writing the above minimization problem (\ref{inner_min_prob}) in standard form, we first introduce variables $w_{l}^+,w_l^-,0\leq l \leq 3,$ satisfying $w_l=w_l^{+}-w_{l}^{-},0\leq l\leq 3$ and change the sign of the inequalities to obtain the following LP problem :

\begin{align}\label{min_standard_form}
  \begin{split}
    &\min_{w_0^{+},w_0^{-},w_1^{+},w_1^{-},w_{2}^{+},w_{2}^{-},w_3^+,w_3^-,z_i}\sum_{i=1}^{N_1}\alpha_i z_i\hspace{0.2cm}\text{subject to}\\
    &z_i+
    (w_0^{+}-w_{0}^{-})+(w_1^{+}-w_1^{-})x_i+(w_2^{+}-w_2^{-})(x_i-K_1)^{+}\\
    &+(w_3^{+}-w_3^{-})(x_i-K_2)^{+}\geq \sum_{j=1}^{N_2}\frac{p_{i,j}c(x_i,y_j)}{\alpha_i},\\
    &z_i
    -(w_0^{+}-w_{0}^{-})-(w_1^{+}-w_1^{-})x_i-(w_2^{+}-w_2^{-})(x_i-K_2)^{+}\\
    &-(w_3^{+}-w_3^{-})(x_i-K_2)^{+}\geq -\sum_{j=1}^{N_2}\frac{p_{i,j}c(x_i,y_j)}{\alpha_i}, \\
    & w_0^{+},w_0^{-},w_1^{+},w_1^{-},w_{2}^{+},w_{2}^{-},w_3^+,w_3^-\geq0, z_i\geq0, i=1,..,N_1.\\  
  \end{split}  
  \end{align}

A straightforward calculation gives the dual maximization LP problem to (\ref{min_standard_form})  as follows :
\begin{align}\label{dual_standard_max_form}
\begin{split}
     &\max_{a_i,b_i}\sum_{i=1}^{N_1}v_i(a_i-b_i)\hspace{0.2cm}\text{subject to}\\
    &\sum_{i=1}^{N_1}(a_i-b_i)=0,\\
    &\sum_{i=1}^{N_1}x_i(a_i-b_i)=0,\\
    &\sum_{i=1}^{N_1}(x_i-K_1)^{+}(a_i-b_i)=0,\\
    &\sum_{i=1}^{N_1}(x_i-K_2)^{+}(a_i-b_i)=0,\\
    &a_i+b_i\leq\alpha_i, i=1,..,N_1,\\
    &a_i,b_i\geq0,i=1,..,N_1, 
\end{split}
\end{align}
where $v_i =\sum_{j=1}^{N_2}\frac{p_{i,j}c(x_i,y_j)}{\alpha_i}$. Using the dual representation (\ref{dual_standard_max_form}) and substituting the value of $v_i$ in terms of the joint probabilities $p_{i,j}$, we can formulate the max-min problem as a single maximization problem.

\begin{section}{Convergence results}\label{Sec:Convergence results}
We begin this section by recalling the definition of the approximating measures $(\mu_n,\nu_n)$ derived using the observable call option prices in Section \ref{section:Framework}. Our aim in this Section is to prove that under certain restrictions the convergence of the given sequences $\{\mu_n\}_{n\in\mathbb{N}},\{\nu_n\}_{n\in\mathbb{N}}$ of probability measures to the true underlying marginal distributions $\mu,\nu$ respectively, implies convergence of the solutions of the corresponding hedging problems given by (\ref{MOT_hedge}). To achieve this, we need some important results from \cite{schmithals2019contributions} and \cite{beiglbock2019dual}.

First, let us define the upper price bound problem for the general market scenario, as described in \cite{schmithals2019contributions} by
\begin{align}\label{primal_d_dim}
    \overline{P}(c):=\sup_{\mathbb{Q}\in\mathcal{M}}\mathbb{E}[c(S_{t_1}^1,\dots,S_{t_n}^1,\dots,S_{t_1}^d,\dots,S_{t_1}^d)]
\end{align}
and the corresponding super hedging problem 
\begin{align}\label{dual_d_dim}
    \overline{D}(c):=\inf \sum_{j=1}^d\sum_{i=1}^n\int_{\mathbb{R}}\varphi_{i,j}(s_{t_i}^j)\mu_{i,j}(ds_{t_i}^j)=\inf\sum_{j=1}^d\sum_{i=1}^n\mathbb{E}_{\mu_{i,j}}[\varphi_{i,j}(s_{t_i}^j)],
    \end{align}
    with
    \begin{align*}
      \varphi_{i,j}\in\mathcal{S}:=\bigg\{u:\mathbb{R}\rightarrow\mathbb{R}\Bigg|u(x)=a+bx+\sum_{l=1}^mc_l(x-k_l)^+,\quad a,b,c_l,k_l\in\mathbb{R},m\in\mathbb{N}\Bigg\}.
    \end{align*}
    The equivalent problems for the standard market case with $d=1$ and $n=2$ (corresponding to our set-up in Sections \ref{section:Framework} and \ref{Section:Hedging problem}) are
    \begin{align}
\label{MOT_1dim}
    P_{2}^c(\mu,\nu):=\sup_{\mathbb{Q}\in\mathcal{M}_2(\mu,\nu)}\mathbb{E}_{\mathbb{Q}}[c(X,Y)]
    \end{align}
    and 
    \begin{align}
        \label{dual_1dim}
        \begin{split}
            D_{2}^c(\mu,\nu)&:=\inf_{(\varphi,\psi,h)\in\mathcal{D}_{2}^{\leq c}}\Bigg\{\int_{\mathbb{R}}\varphi(x)\mu(dx)+\int_{\mathbb{R}}\psi(y)\nu(dy)\Bigg\}\\
            &=\inf_{(\varphi,\psi,h)\in\mathcal{D}_{2}^{\leq c}}\{\mathbb{E}_{\mu}[\varphi(X)]+\mathbb{E}_{\nu}[\psi(Y)].
        \end{split}
    \end{align}
    with
    \begin{align*}
      \mathcal{D}_{2}^{\leq c};=\{(\varphi,\psi,h)|\varphi^+\in\mathbb{L}^1(\mathbb{R},\mu),\psi^+\in\mathbb{L}^1(\mathbb{R},\nu),h\in\mathbb{L}^0(\mathbb{R}),\\
      \varphi(x)+\psi(y)+h(x)(y-x)\geq c(x,y),(x,y)\in\mathbb{R}^2\} . 
    \end{align*}
 Corollary \ref{Corollary_upper_semicont} gives a sufficient condition for the equivalence of the problems (\ref{primal_d_dim}) and (\ref{dual_d_dim}) as follows :

\begin{cor}(\cite{schmithals2019contributions})
\label{Corollary_upper_semicont}
  Let $\mathcal{M}\neq \phi$ and $c:\mathbb{R}^{nd}\rightarrow[-\infty,\infty)$ be an upper semi-continuous payoff function such that there is a constant $K\in\mathbb{R}$ with
  \begin{align*}
      c(s_{t_1}^1,..,s_{t_n}^{1},\cdots,s_{t_1}^d,\cdots,s_{t_n}^d)\leq K(1+\sum_{j=1}^{d}\sum_{i=1}^{n}|s_{t_i}^{j}|)
  \end{align*}
  for all $(s_{t_1}^1,..,s_{t_n}^{1},\cdots,s_{t_1}^d,\cdots,s_{t_n}^d)\in\mathbb{R}^{nd}$. Then $\overline{P}(c)=\overline{D}(c)$ and there is a $\mathbb{Q}^*\in\mathcal{M}$ such that $\overline{P}(c)=\mathbb{E}_{\mathbb{Q}^*}[c].$
\end{cor}

The \textit{dual minimizer}, which is the super hedge, may not be attained, and hence, we need certain conditions to ensure the existence of a \textit{dual minimizer}. This leads us to the following.

\begin{defn}(\cite{beiglbock2019dual})
\label{def_concavifier}
    Let $\mu\leq_{c}\nu$ and $c:\mathbb{R}^2\rightarrow\mathbb{R}$ be a payoff function. Then a triple $(\varphi,\psi,h)$ of functions $\varphi:\mathbb{R}\rightarrow\mathbb{R}\cup\{\infty\},\psi:\mathbb{R}\rightarrow\mathbb{R}\cup\{\infty\}$ and $h:\mathbb{R}\rightarrow\mathbb{R}$ is called \textit{dual minimizer}, if $\varphi$ is finite $\nu$-almost surely and, for any maximizer $\mathbb{Q}^*\in\mathcal{M}_2(\mu,\nu)$ of the upper price bound problem in (\ref{MOT_1dim}), we have
    \begin{align*}
        \varphi(x)+\psi(y)+h(x)(y-x)\geq c(x,y),\quad \forall (x,y)\in\mathbb{R}^2\\
        \varphi(x)+\psi(y)+h(x)(y-x)= c(x,y),\quad \text{for}\quad \mathbb{Q}^*-\text{almost every}\quad(x,y).
    \end{align*}
    \end{defn}
    \begin{defn}
        \cite{beiglbock2019dual}
        Let $J$ be an interval and $\mu\in\mathcal{P}_{\alpha}(\mathbb{R}).$ We say that a function $c:\mathbb{R}^2\rightarrow\mathbb{R}$ is \textit{semi-concave} in $y\in J$ $\mu-$uniformly, if there exists a Borel function $u:J\rightarrow\mathbb{R}$ such that for $\mu$-almost every $x$, the mapping $y\mapsto c(x,y)+u(y)$ is continuous and concave on $J.$ In this case, we say that $u$ is a $y$-\textit{concavifier} on $J$ for $c.$
    \end{defn}
    Theorems \ref{thm_dual_min_and_concavifier} and \ref{thm_lipschitz_bound} give conditions on the cost function that guarantee the existence of a dual minimizer.
    \begin{thm}
        \cite{beiglbock2019dual}
        \label{thm_dual_min_and_concavifier}
        Let $\mu\leq_c \nu, J:=\text{conv(supp}(\nu))$ and $c:\mathbb{R}^2\rightarrow\mathbb{R}$. Suppose that there a $y-$concavifier $u$ exists on $J$ for $c$. If $J$ is not compact, then further suppose $y\rightarrow c(x,y)+u(y)$ is of linear growth on $J$. Then a dual minimizer exists in the sense of Definition \ref{def_concavifier}.
    \end{thm}
    \begin{thm}\cite{beiglbock2019dual}
        \label{thm_lipschitz_bound}
        Suppose the assumptions of Theorem \ref{thm_dual_min_and_concavifier} are satisfied and that further $c$ is Lipschitz continuous on $J\times J$ and $u$ is Lipschitz continuous on $J$. Then there exists a dual minimizer $(\varphi,\psi,h)$ such that $\varphi$ and $\psi$ are Lipschitz continuous on $J$ and $|h|$ is bounded on $J$.
    \end{thm}
    The following remark from \cite{schmithals2019contributions} gives specific Lipschitz bounds for the dual minimizers depending on the domain of the cost function.
    \begin{rem}
        (\cite{schmithals2019contributions})
        \label{remark_schmithals_lipschitz_bound}
        \begin{enumerate}
            \item If $c$ and $u$ in Theorem \ref{thm_lipschitz_bound} are Lipschitz continuous with constant $\Lambda$, then the dual minimizer may be chosen such that $\varphi$ and $\psi$ are Lipschitz continuous with constants $19\Lambda$ and $17\Lambda$ on $J$, and $|h|$ is bounded by $18\Lambda$ on $J$. This is computed in the proof of \cite{beiglbock2019dual}(Theorem 2.5).
            \item In a former version, the authors prove Theorem \ref{thm_lipschitz_bound} for compact $J$. Then the proof yields that the dual minimizer may be chosen such that $\varphi$ and $\psi$ are Lipschitz continuous with constants $7\Lambda$ and $5\Lambda$ on $J$, and $|h|$ is bounded by $6\Lambda$ on $J$.
            \item Analyzing the proof of \cite{beiglbock2019dual}(Theorem 2.5), the authors recognize that the global Lipschitz condition may be weakened. Instead, one demands that there is a $\Lambda>0$ such that for the domain $(I,J)$ of every irreducible component of $(\mu,\nu)$, we have
            \begin{itemize}
                \item $c_y(x,b-)+u'(b-)-c_y(x,a+)-u'(a+)\leq4\Lambda\quad \forall x\in I=(a,b).$
                \item $|c(x,y)-c(x',y)|\leq \Lambda|x-x'|\quad\forall x,x',y\in J.$
            \end{itemize}
        \end{enumerate}
    \end{rem}
 The following convergence result from \cite{schmithals2019contributions} provides a bound on the pricing error due to the availability of option prices at finitely many equally spaced strikes over the bounded support of the underlying measure.
    
    \begin{thm}\cite{schmithals2019contributions}
    \label{Thm_schmithals_convergence}
       Let $(\mu,\nu)\in \mathcal{P}_{K.L}^{\leq c}$. Let $c:\mathbb{R}_{+}^{2}\rightarrow\mathbb{R}$ be a Lipschitz continuous payoff function such that $c_{yy}$ exists. We denote by $\hat{\Lambda}$ the Lipschitz constant of $c$ and assume $\max\{\hat{\Lambda},\sup_{(x,y)\in\mathbb{R}^2}|c_{yy}(x,y)|\}\leq\Lambda$. then, for any $n\in \mathbb{N}$, we have
       \begin{align*}
           \bigg| \sup_{\mathbb{Q}\in\mathcal{M}_2(\mu_{n}^d,\nu_{n}^d)}\mathbb{E}_{\mathbb{Q}}[c(X,Y)]-\sup_{\mathbb{Q}\in\mathcal{M}_2(\mu,\nu)}\mathbb{E}_{\mathbb{Q}}[c(X,Y)]\bigg|\leq\frac{M_c}{2^n},
       \end{align*}
       where $M_c=(7K+5L).\tilde{\Lambda}$ with $\tilde{\Lambda}=\Lambda.\max\{L,1\}$. If we additionally suppose that $C_{\mu},C_{\nu}\in\mathcal{ C}^2(\mathbb{R}_{+}),$ then, for any $n\in\mathbb{N},$ we have
       \begin{align*}
           \bigg| \sup_{\mathbb{Q}\in\mathcal{M}_2(\mu_{n}^d,\nu_{n}^d)}\mathbb{E}_{\mathbb{Q}}[c(X,Y)]-\sup_{\mathbb{Q}\in\mathcal{M}_2(\mu,\nu)}\mathbb{E}_{\mathbb{Q}}[c(X,Y)]\bigg|\leq\frac{M_d}{2^{n+1}},
       \end{align*}
       where $M_d=(7T_{\mu}K^2+5T_{\nu}L^2).\tilde{\Lambda}$ with $T_{\mu}=\sup_{\kappa\in[0,K]}|C_{\mu}^{"}(\kappa)|$ and $T_{\nu}=\sup_{\lambda\in[0,L]}|C_{\nu}^{"}(\lambda)|$.
    \end{thm}

   \textbf{ This brings us to our main result. We use similar techniques to the ones used in the proof of Theorem \ref{Thm_schmithals_convergence}. }
 \begin{thm}
    \label{Our_thm_time_T_convergence}
        Let $c:\mathbb{R}_{+}^{2}\rightarrow\mathbb{R}$ be a Lipschitz continuous payoff function.
       \begin{enumerate}
          \item  We further assume that for fixed $w_0,w_1,\dots,w_M\in\mathbb{R}$ there exists a Lipschitz function $u:\mathbb{R}_{+}\rightarrow\mathbb{R}$ such that $y\rightarrow|c(x,y)-\{w_0+\sum_{i=1}^{M}w_i(x-K_i)^+\}|+u(y)$ is concave on $\mathbb{R}_{+}$ for $\mu-$almost every $x\in \mathbb{R}_{+}$. Let $\hat{\Lambda}$ and $\hat{\Theta}$ denote the Lipschitz constants of $c$ and $u$ respectively and assume $\max\{\hat{\Lambda}+\sum_{i=1}^M |w_i|,\hat{\Theta}\}\leq\Lambda_w$. Then, for any $n\in \mathbb{N}$, we have
       \begin{align}
       \label{1st_convergence_inequality_unbdd_case}
          \begin{split}
             &\Bigg|\sup_{\mathbb{Q}\in\mathcal{M}_2(\mu_{n}^d,\nu_{n}^d)}\mathbb{E}_{\mathbb{Q}}\bigg[|c(X,Y)-\{w_0+\sum_{i=1}^{M}w_i(x-K_i)^+\}|\bigg]-\\
             &\sup_{\mathbb{Q}\in\mathcal{M}_2(\mu,\nu)}\mathbb{E}_{\mathbb{Q}}\bigg[|c(X,Y)-\{w_0+\sum_{i=1}^{M}w_i(x-K_i)^+\}|\bigg]\Bigg|\\
          &\leq19\Lambda_w W(\mu,\mu_{n}^d)+17\Lambda_w W(\nu,\nu_{n}^d). 
          \end{split}
       \end{align}
       Further, if the weights $w=(w_0,w_1,..,w_M)$ are restricted over a compact set $A\subset\mathbb{R}^{M+1}$ we have
       \begin{align}\label{ineq_unbdd_case}
          \begin{split}
             &\Bigg|\inf_w\sup_{\mathbb{Q}\in\mathcal{M}_2(\mu_{n}^d,\nu_{n}^d)}\mathbb{E}_{\mathbb{Q}}\bigg[|c(X,Y)-\{w_0+\sum_{i=1}^{M}w_i(x-K_i)^+\}|\bigg]\\
             &-  \inf_w\sup_{\mathbb{Q}\in\mathcal{M}_2(\mu,\nu)}\mathbb{E}_{\mathbb{Q}}\bigg[|c(X,Y)-\{w_0+\sum_{i=1}^{M}w_i(x-K_i)^+\}|\bigg]\Bigg|\\
          &\leq B\times( W(\mu,\mu_{n}^d)+ W(\nu,\nu_{n}^d)), 
          \end{split}
       \end{align}
       where $B = \max\{19\Lambda_w,17\Lambda_w\}$.
       \item Let us assume in addition that $(\mu,\nu)\in \mathcal{P}_{K.L}^{\leq c}$ and for fixed $w_0,w_1,\dots,w_M\in\mathbb{R}$ there exists a Lipschitz function $u:[0,L]\rightarrow\mathbb{R}$ such that $y\rightarrow|c(x,y)-\{w_0+\sum_{i=1}^{M}w_i(x-K_i)^+\}|+u(y)$ is concave on $[0,L]$ for $\mu-$almost every $x\in \mathbb{R}_{+}$. Let $\hat{\Lambda}$ and $\hat{\Theta}$ denote the Lipschitz constants of $c$ and $u$ respectively and assume $\max\{\hat{\Lambda}+\sum_{i=1}^M |w_i|,\hat{\Theta}\}\leq\Lambda_w$. Then, for any $n\in \mathbb{N}$, we have
       \begin{align}
        \label{convergence_inequality_bdd_case}
          \begin{split}
             &\Bigg|\sup_{\mathbb{Q}\in\mathcal{M}_2(\mu_{n}^d,\nu_{n}^d)}\mathbb{E}_{\mathbb{Q}}\bigg[|c(X,Y)-\{w_0+\sum_{i=1}^{M}w_i(x-K_i)^+\}|\bigg]\\
             &-  \sup_{\mathbb{Q}\in\mathcal{M}_2(\mu,\nu)}\mathbb{E}_{\mathbb{Q}}\bigg[|c(X,Y)-\{w_0+\sum_{i=1}^{M}w_i(x-K_i)^+\}|\bigg]\Bigg|\\
          &\leq7\Lambda_w W(\mu,\mu_{n}^d)+5\Lambda_w W(\nu,\nu_{n}^d). 
          \end{split}
       \end{align}
       Further, if the weights $w=(w_0,w_1,..,w_M)$ are restricted over a compact set $A\subset\mathbb{R}^{M+1}$ we have
       \begin{align}\label{ineq_bdd_case}
          \begin{split}
             &\Bigg|\inf_w\sup_{\mathbb{Q}\in\mathcal{M}_2(\mu_{n}^d,\nu_{n}^d)}\mathbb{E}_{\mathbb{Q}}\bigg[|c(X,Y)-\{w_0+\sum_{i=1}^{M}w_i(x-K_i)^+\}|\bigg]\\
             &-  \inf_w\sup_{\mathbb{Q}\in\mathcal{M}_2(\mu,\nu)}\mathbb{E}_{\mathbb{Q}}\bigg[|c(X,Y)-\{w_0+\sum_{i=1}^{M}w_i(x-K_i)^+\}|\bigg]\Bigg|\\
          &\leq B\times( W(\mu,\mu_{n}^d)+ W(\nu,\nu_{n}^d)), 
          \end{split}
       \end{align}
       where $B=\max\{7\Lambda_w,5\Lambda_w\}$. 
       \end{enumerate}
    \end{thm}
  \begin{proof}
       
   Since $c(x,y)$ is Lipschitz, it readily follows that the absolute hedging error function $|c(x,y)-\{w_0+\sum_{i=1}^{M}w_i(x-K_i)^+\}|$ is also Lipschitz continuous on $\mathbb{R}^2_{+}=\text{conv(supp(}\nu))\times\text{conv(supp(}\nu))$ with Lipschitz constant $\hat{\Lambda}+\sum_{i=1}^{M}|w_i|$. Following the proof of Theorem \ref{Thm_schmithals_convergence} from \cite{schmithals2019contributions}, we apply Corollary \ref{Corollary_upper_semicont} to get
    \begin{align*}
       &\sup_{\mathbb{Q}\in\mathcal{M}_2(\mu_{n}^d,\nu_{n}^d)}\mathbb{E}_{\mathbb{Q}}\bigg[|c(X,Y)-\{w_0+\sum_{i=1}^{M}w_i(x-K_i)^+\}|\bigg]\\
       &=\inf_{(\varphi,\psi,h)\in\mathcal{D}_{2}^{\geq c}}\Bigg\{\int_{\mathbb{R}_+}\varphi(x)\mu_{n}^d(dx)+\int_{\mathbb{R}_+}\psi(x)\nu_{n}^d(dy)\Bigg\},
       \end{align*}
       and
       \begin{align*}
       &\sup_{\mathbb{Q}\in\mathcal{M}_2(\mu,\nu)}\mathbb{E}_{\mathbb{Q}}\bigg[|c(X,Y)-\{w_0+\sum_{i=1}^{M}w_i(x-K_i)^+\}|\bigg]\\
       &=\inf_{(\varphi,\psi,h)\in\mathcal{D}_{2}^{\geq c}}\Bigg\{\int_{\mathbb{R}_+}\varphi(x)\mu(dx)+\int_{\mathbb{R}_+}\psi(x)\nu(dy)\Bigg\} .
    \end{align*}
    Under the assumption that there exists a Lipschitz function $u:\mathbb{R}_{+}=\text{conv(supp(}\nu))\rightarrow\mathbb{R}$ such that $y\rightarrow|c(x,y)-\{w_0+\sum_{i=1}^{M}w_i(x-K_i)^+\}|+u(y)$ is concave on $\mathbb{R}_+$ for $\mu-$almost every $x\in \mathbb{R}_{+}$, we can apply Theorem \ref{thm_lipschitz_bound}.  This gives us solutions $(\varphi^*,\psi^*,h^*)$ and $(\varphi_{n}^*,\psi_{n}^*,h_{n}^*)$ for the dual problems with respect to $(\mu,\nu)$ and $(\mu_{n}^d,\nu_{n}^d)$ respectively. Applying Remark \ref{remark_schmithals_lipschitz_bound}, $\varphi^*$ and $\varphi_{n}^*$ are Lipschitz continuous with constant $19\Lambda_w$, and $\psi^*$ and $\psi_{n}^*$ are Lipschitz continuous with constant $17\Lambda_w$ where $\Lambda_w$ is the maximum of the Lipschitz constants of $c(x,y)-\{w_0+\sum_{i=1}^{M}w_i(x-K_i)^+\}|$ and $u(y)$. This gives
    \begin{align*}
&\sup_{\mathbb{Q}\in\mathcal{M}_2(\mu_{n}^d,\nu_{n}^d)}\mathbb{E}_{\mathbb{Q}}\bigg[|c(X,Y)-\{w_0+\sum_{i=1}^{M}w_i(x-K_i)^+\}|\bigg]\\
&-  \sup_{\mathbb{Q}\in\mathcal{M}_2(\mu,\nu)}\mathbb{E}_{\mathbb{Q}}\bigg[|c(X,Y)-\{w_0+\sum_{i=1}^{M}w_i(x-K_i)^+\}|\bigg]\\
      &=\inf_{(\varphi,\psi,h)\in\mathcal{D}_{2}^{\geq c}}\Bigg\{\int_{\mathbb{R}_+}\varphi(x)\mu_{n}^d(dx)+\int_{\mathbb{R}_+}\psi(x)\nu_{n}^d(dy)\Bigg\}\\
      &-\inf_{(\varphi,\psi,h)\in\mathcal{D}_{2}^{\geq c}}\Bigg\{\int_{\mathbb{R}_+}\varphi(x)\mu(dx)+\int_{\mathbb{R}_+}\psi(x)\nu(dy)\Bigg\}\\
      &\leq \int_{\mathbb{R}_+}\varphi^*(x)\mu_{n}^d(dx)+\int_{\mathbb{R}_+}\psi^*(x)\nu_{n}^d(dy)\\
      &-\bigg(\int_{\mathbb{R}_+}\varphi^*(x)\mu^d(dx)+\int_{\mathbb{R}_+}\psi^*(x)\nu^d(dy)\bigg)\\
      &=\int_{\mathbb{R}_+}\varphi^*(x)(\mu_{n}^d-\mu)(dx)+\int_{\mathbb{R}_+}\psi^*(x)(\nu_{n}^d-\nu)(dy)\\
      &\leq19\Lambda_w W(\mu,\mu_{n}^d)+17\Lambda_w W(\nu,\nu_{n}^d).
    \end{align*}
    Analogously using $\varphi_{n}^*$ and $\psi_{n}^*$ in the first inequality instead of $\varphi^*$ and $\psi^*$ one obtains
    \begin{align*}
       & \sup_{\mathbb{Q}\in\mathcal{M}_2(\mu,\nu)}\mathbb{E}_{\mathbb{Q}}\bigg[|c(X,Y)-\{w_0+\sum_{i=1}^{M}w_i(x-K_i)^+\}|\bigg]\\
       &-\sup_{\mathbb{Q}\in\mathcal{M}_2(\mu_{n}^d,\nu_{n}^d)}\mathbb{E}_{\mathbb{Q}}\bigg[|c(X,Y)-\{w_0+\sum_{i=1}^{M}w_i(x-K_i)^+\}|\bigg] \\
        &\leq19\Lambda_w W(\mu,\mu_{n}^d,)+17\Lambda_w W(\nu,\nu_{n}^d).
    \end{align*}

   This proves (\ref{1st_convergence_inequality_unbdd_case}). To prove the second part of the statement $1$ in Theorem \ref{Our_thm_time_T_convergence}, let 
    \begin{align*}
        f_n(w)= \sup_{\mathbb{Q}\in\mathcal{M}_2(\mu_{n}^d,\nu_{n}^d)}\mathbb{E}_{\mathbb{Q}}\bigg[|c(X,Y)-\{w_0+\sum_{i=1}^{M}w_i(x-K_i)^+\}|\bigg],
    \end{align*} 
    and 
    \begin{align*}
     f(w)=\sup_{\mathbb{Q}\in\mathcal{M}_2(\mu,\nu)}\mathbb{E}_{\mathbb{Q}}\bigg[|c(X,Y)-\{w_0+\sum_{i=1}^{M}w_i(x-K_i)^+\}|\bigg]    
    \end{align*} 
    where $w=(w_0,w_1,\dots,w_M)\in A$. This gives
    \begin{align*}
        &f_n(w)-f(w)\leq|f_n(w)-f(w)|\leq \sup_w |f_n(w)-f(w)|,\\
        &\implies f_n(w)\leq f(w)+\sup_w |f_n(w)-f(w)|,\\
        &\implies \inf_w f_n(w)\leq \inf_w f(w) + \sup_w |f_n(w)-f(w)|,\\
        &\implies \inf_w f_n(w)-\inf_wf(w)\leq \sup_w |f_n(w)-f(w)|\\
        &\leq \sup_w (19\Lambda_w W(\mu,\mu_{n}^d)+17\Lambda_w W(\nu,\nu_{n}^d))\\&= B\times( W(\mu,\mu_{n}^d)+W(\nu,\nu_{n}^d)),
    \end{align*}
  where $B=\sup_w\{19\Lambda_w,17\Lambda_w\}$, with the supremum being taken over the compact set over which $w$ takes values. Since $\Lambda_w$ is the maximum of the Lipschitz constants of 
$c(x,y)-\{w_0+\sum_{i=1}^{M}w_i(x-K_i)^+\}|$ and $u(y)$, the supremum of these over a compact set is also a finite value.

   Similarly, we obtain 
   \begin{align*}
        \inf_wf(w) - \inf_w f_n(w)\leq B\times( W(\mu,\mu_{n}^d)+W(\nu,\nu_{n}^d)),
   \end{align*}
   which gives us the desired inequality (\ref{ineq_unbdd_case}). The proof of (\ref{ineq_bdd_case}) is similar and hence we omit it.\end{proof}

   \end{section}
\section{Numerical Results}\label{Sec:Numerical examples}
In this section, we focus on highlighting the utility of our \textit{min-max} hedging approach (\ref{MoMOT_hedge_t1}) when the option prices are generated using a \textit{Black-Scholes} (BS) and a \textit{Merton Jump diffusion}(MJD) model.

\subsection{Black Scholes Model} \label{subsec:BS_exp}
\subsubsection{Asian Option}
We consider an Asian option with payoff $c(X,Y)=(\frac{1}{2}(X+Y)-K)^+$. The parameters are : $ S_0=1=K,\sigma=0.2,\mu=0=r,t_1=0.5,T=1$. The price of the option obtained using Monte Carlo simulations with $ 10^5$ paths is $0.06284$.

We use a non-uniform grid of $12$ discretization points [0,         0.65985287, 0.69305573, 0.83860362, 0.86371482, 0.97595447,1.00102542,1.0484879,  1.09459717, 1.15062857, 1.57436388, 2] centered around the spot price $S_0$ generated from a normal distribution with variance $\sigma=0.2$ and mean $S_0$. The corresponding option prices for maturities $t_1$ and $T$ are used to calculate the discretely supported marginal distributions.

In figure \ref{fig:cond_plots_BSAsian} we plot the conditional value of the target option at $t_1$  using the joint probabilities $p_{i,j}$ obtained using the min-max problem for the different discretization points $x_i,i=0,..,M$ at $t_1$ given in the x-axis, denoted by the blue line. The true BS prices of the options at these discretization points are given by the orange lines, and the green lines denote the hedging portfolio values. 
\begin{figure}
     \centering
     \begin{subfigure}[b]{0.4\textwidth}
         \centering
         \includegraphics[width=\textwidth]{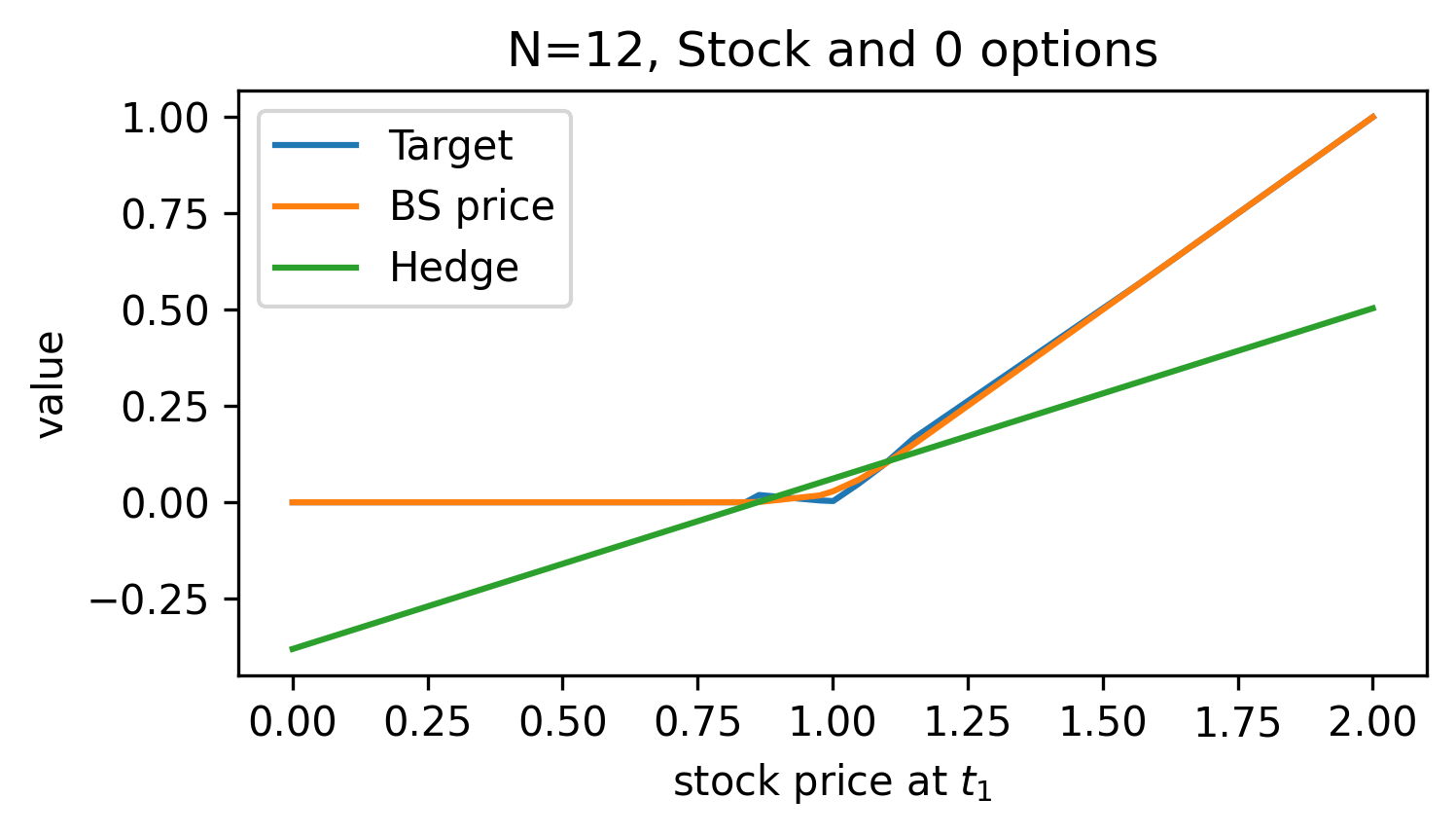}
         \caption{Plot for only stock }
         \label{fig:1option_condplot_bs_asian}
     \end{subfigure}
\hfill
     \begin{subfigure}[b]{0.4\textwidth}
         \centering
         \includegraphics[width=\textwidth]{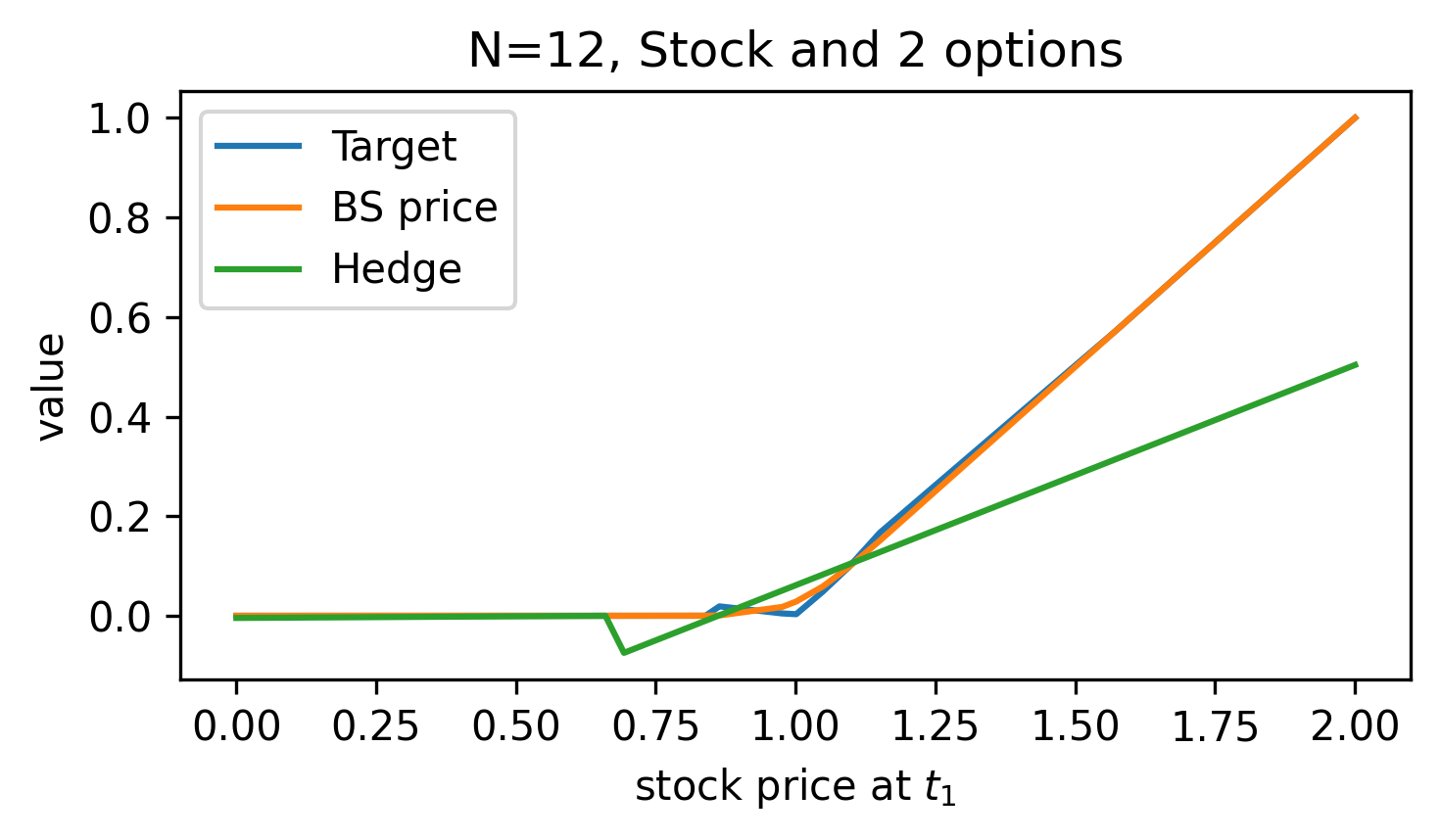}
         \caption{Plot for stock and 2 options}
         \label{fig:2option_condplot_bs_asian}
     \end{subfigure}
\hfill
     \begin{subfigure}[b]{0.4\textwidth}
         \centering
         \includegraphics[width=\textwidth]{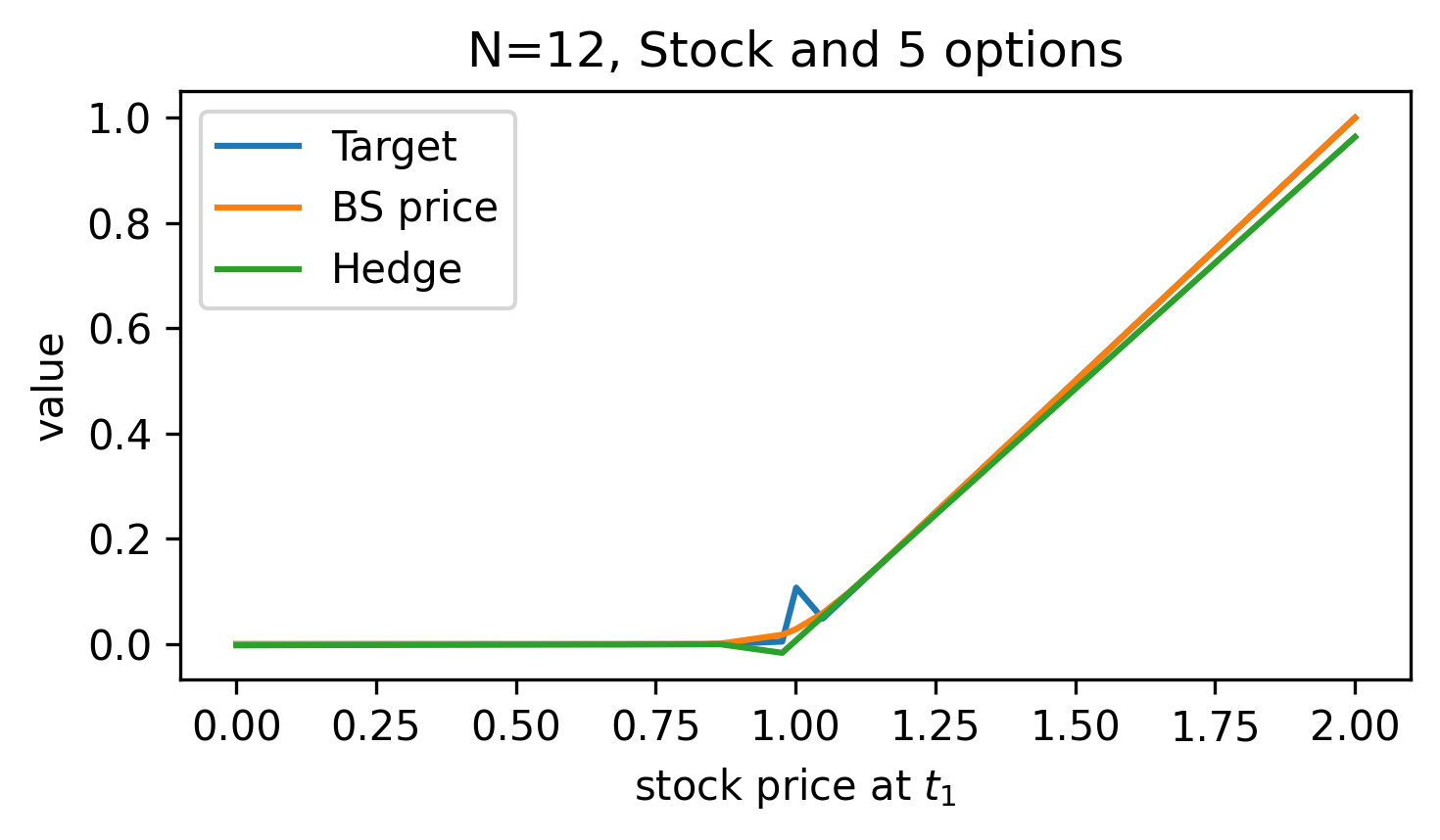}
         \caption{Plot for stock and 5 options}
         \label{fig:6option_condplot_bs_asian}
     \end{subfigure}
     \hfill
     \begin{subfigure}[b]{0.4\textwidth}
         \centering
         \includegraphics[width=\textwidth]{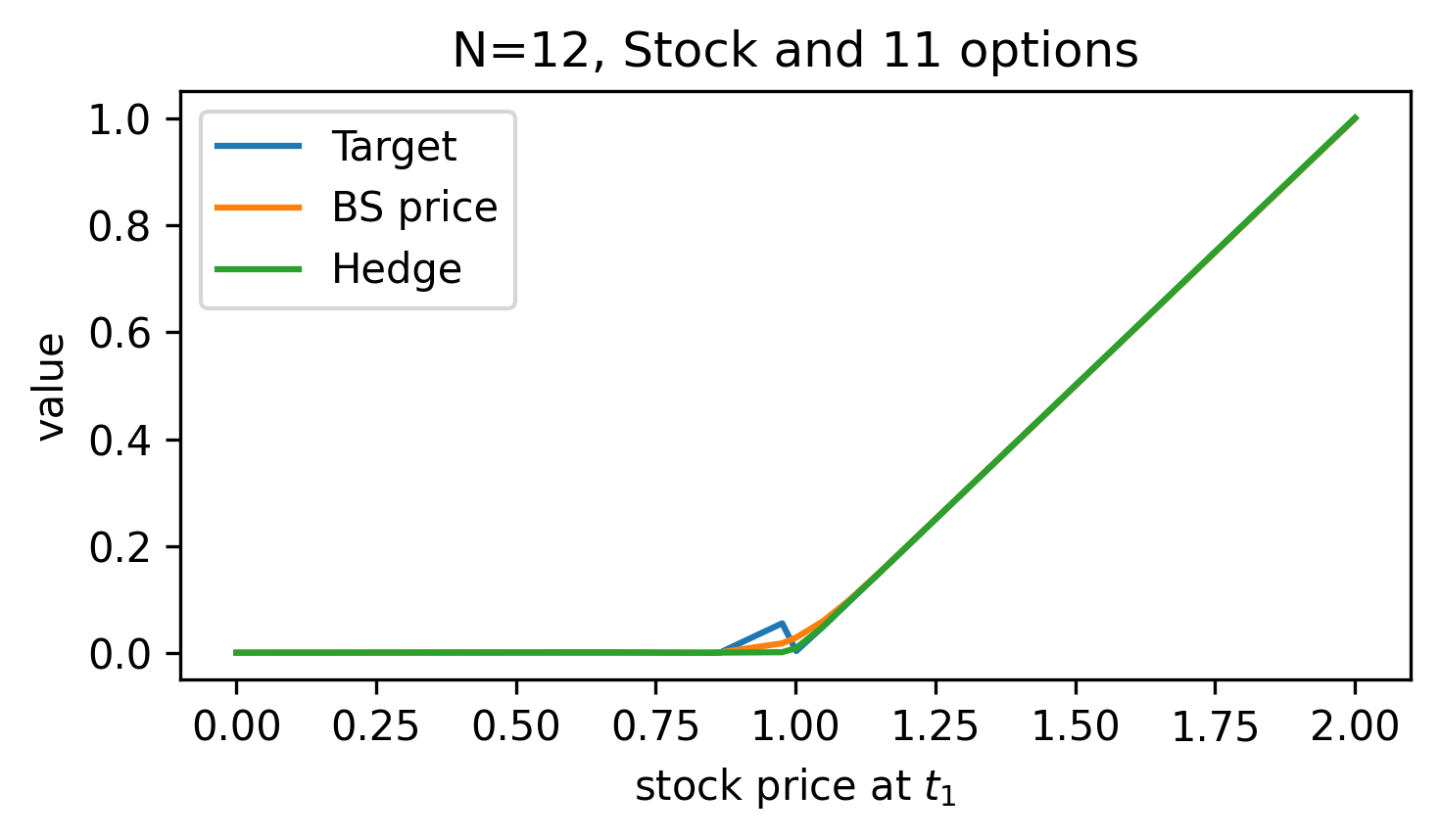}
         \caption{Plot for stock and 11 options}
         \label{fig:12option_condplot_bs_asian}
     \end{subfigure}
        \caption{Plots of the conditional value of the target option under the worst possible scenarios and the corresponding hedging portfolio values for an increasing number of options for the Asian option under the Black Scholes model.}
        \label{fig:cond_plots_BSAsian}
\end{figure}

We observe that the conditional target option value obtained using the $p_{i,j}$ and the true call option price under the Black Scholes model mostly coincide for all cases. The hedging portfolio value starts aligning with the conditional target option value with an increasing number of options, starting from $6$ options. This indicates that holding as few as $6$ options in the hedging portfolio for the given choice of strikes and parameters provides a considerable reduction in the hedging error.

\subsubsection{Forward Start Option}
\textbf{Effect of increasing number of options in the hedging portfolio :}
We consider a forward start option with payoff $c(X,Y)=(Y-X)^+$. The parameters are : $S_0=1,\sigma=0.2,\mu=0=r,t_1=0.5,T=1$. The price of the option obtained using Monte Carlo simulations with $ 10^5$ paths is $ 0.05647$.

We use the same non-uniform grid of $12$ discretization points to calculate the marginal distributions.

The plots in Figure \ref{fig:cond_plots_BSForward} show that neither the hedging portfolio (green line) nor the target option price obtained using the worst case probabilities (blue line) align with the true price given by the original line. However, the addition of more options ( $\geq4$)  yields a better fit for the target option to the true price than with fewer options ($\leq3)$.
\begin{figure}
     \centering
     \begin{subfigure}[b]{0.4\textwidth}
         \centering
         \includegraphics[width=\textwidth]{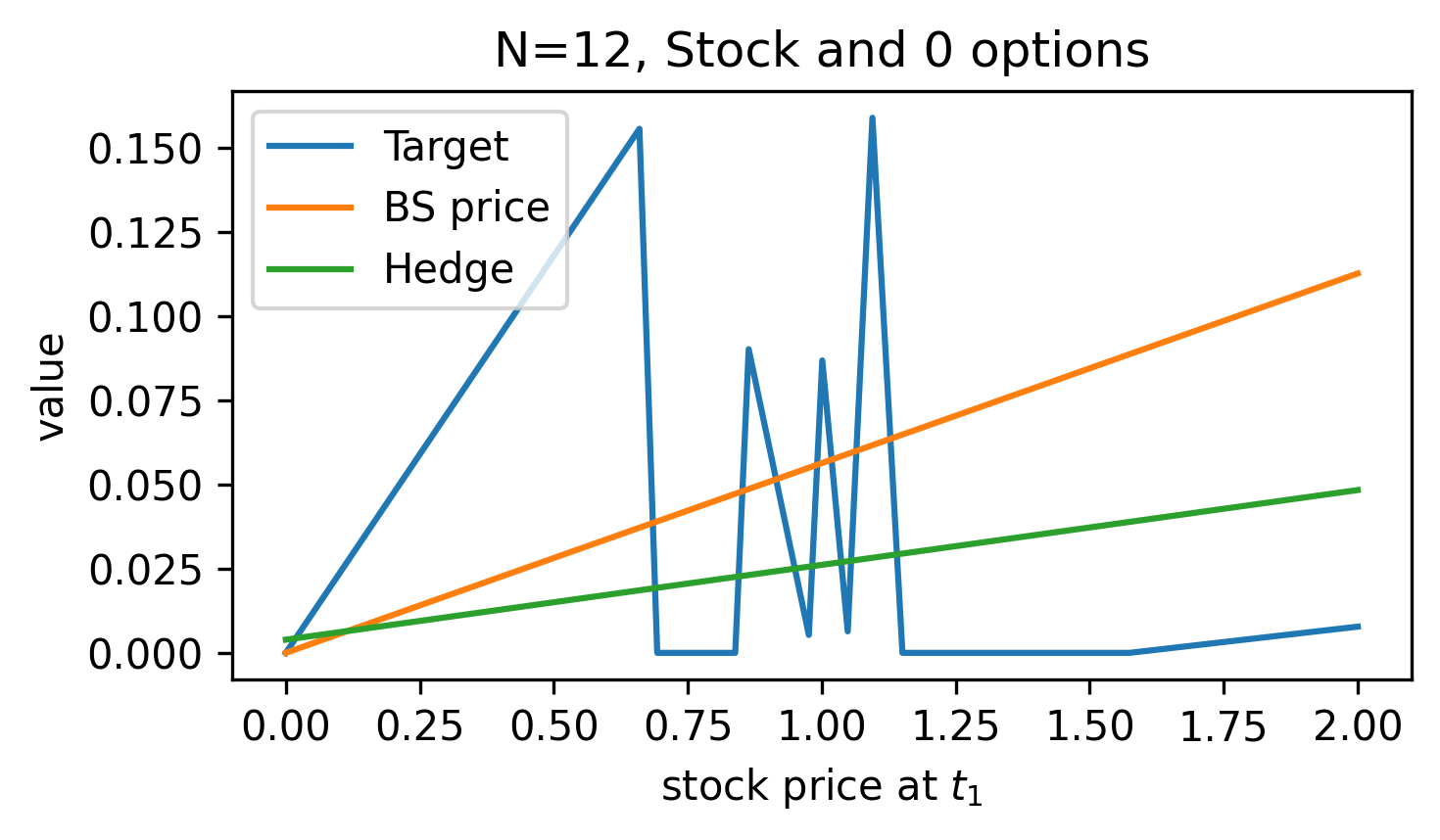}
         \caption{Plot for 1 option}
         \label{fig:1option_condplot_bs_forward}
     \end{subfigure}
     \hfill
     \begin{subfigure}[b]{0.4\textwidth}
         \centering
         \includegraphics[width=\textwidth]{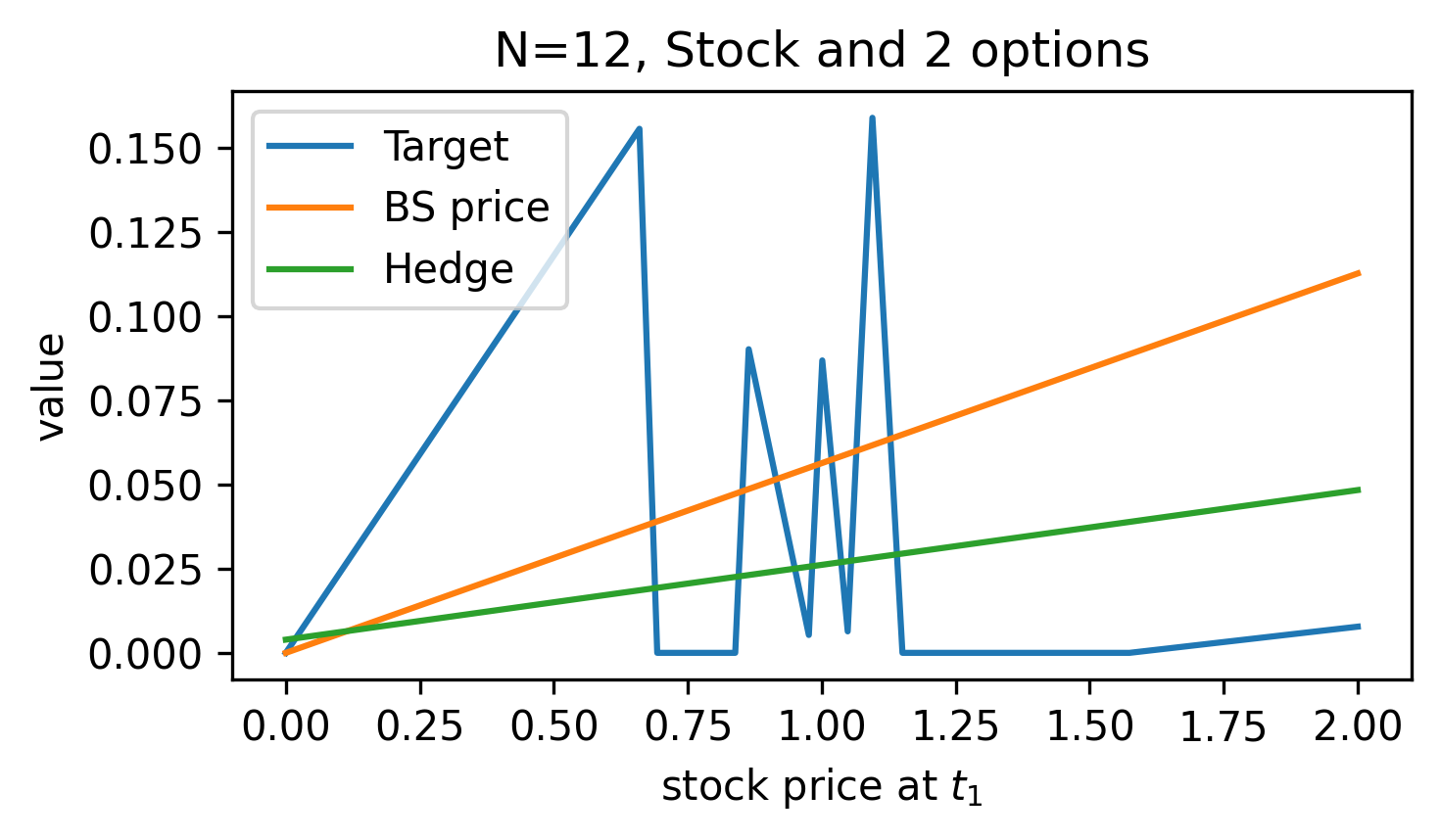}
         \caption{Plot for 3 options}
         \label{fig:3option_condplot_bs_forward}
     \end{subfigure}
     \hfill
     \begin{subfigure}[b]{0.4\textwidth}
         \centering
         \includegraphics[width=\textwidth]{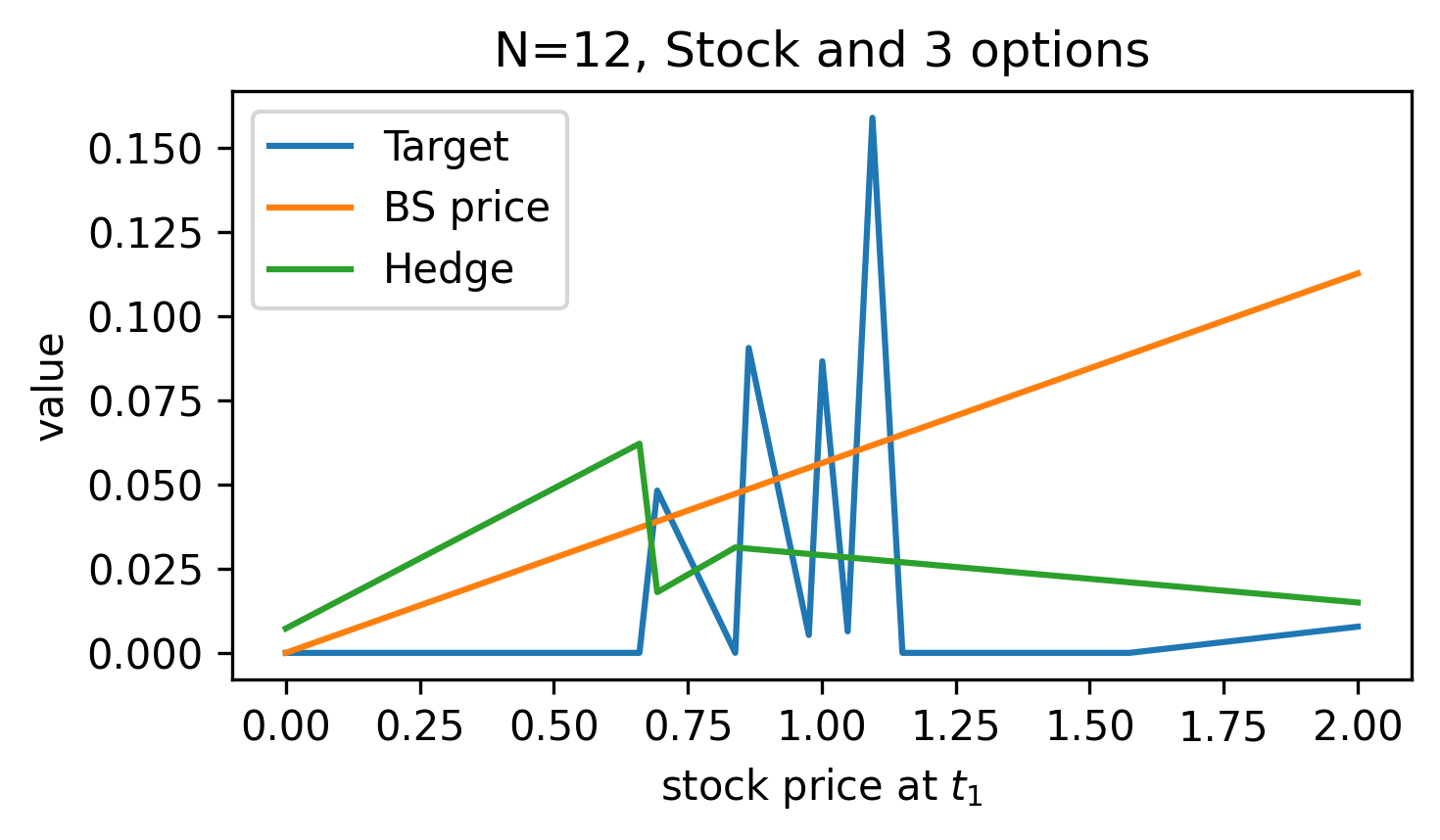}
         \caption{Plot for 4 options}
         \label{fig:4option_condplot_bs_forward}
     \end{subfigure}
     \hfill
     \begin{subfigure}[b]{0.4\textwidth}
         \centering
         \includegraphics[width=\textwidth]{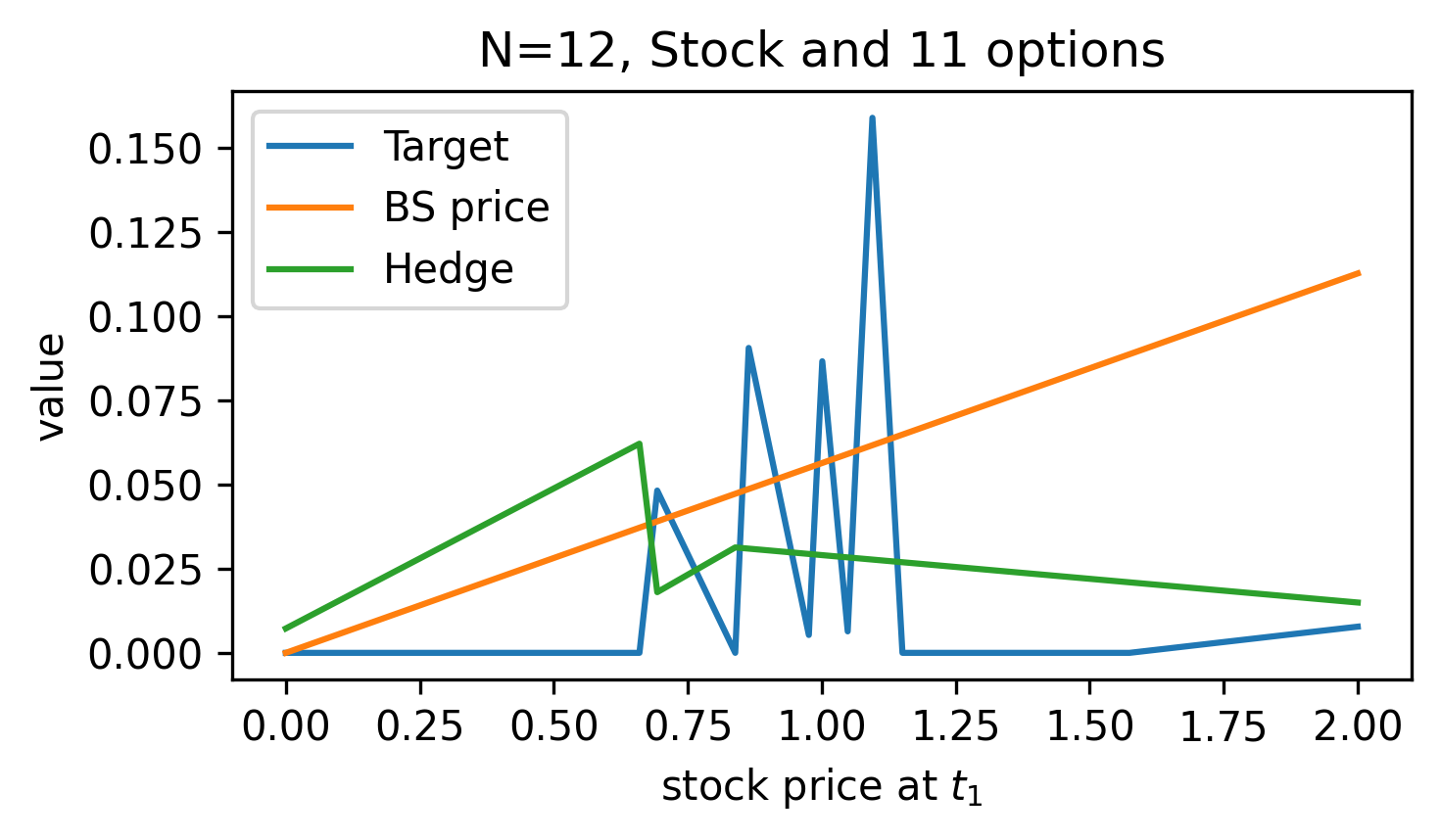}
         \caption{Plot for 12 options}
         \label{fig:12option_condplot_bs_forward}
     \end{subfigure}
        \caption{Plots of the conditional value of the target option under the worst possible scenarios and the corresponding hedging portfolio values for an increasing number of options for the Forward start option under the Black Scholes model.}
        \label{fig:cond_plots_BSForward}
\end{figure}

A natural question then is to study the performance obtained using the dual superhedging approach (\ref{dual_1dim}) and compare it with our min-max hedge throughout the duration of the hedge until maturity $t_1$. We postpone this till subsection \ref{subsec:forward_BS_Tmaturity_simulations}, where we use simulated stock paths to study the performance of the respective hedging algorithms.

\subsection{Merton Jump Diffusion Model}

Since we obtain similar results to those for the BS model scenario, we restrict our attention here to the forward start option only.

We consider a forward start option with payoff $c(X,Y)=(Y-X)^+$. The parameters are : $S_0=1,\sigma=0.2,\mu=0=r,\sigma=0.14,\mu_j=-0.1,\sigma_j=0.13,t_1=0.5,T=1$. The price of the option obtained using Monte Carlo simulations with $ 10^5$ paths is $ 0.07063$. We use the same non-uniform grid of $12$ discretization points as for BS model to calculate the marginal distributions.

\subsubsection{Forward Start Option}
\begin{figure}
     \centering
     \begin{subfigure}[b]{0.4\textwidth}
         \centering
         \includegraphics[width=\textwidth]{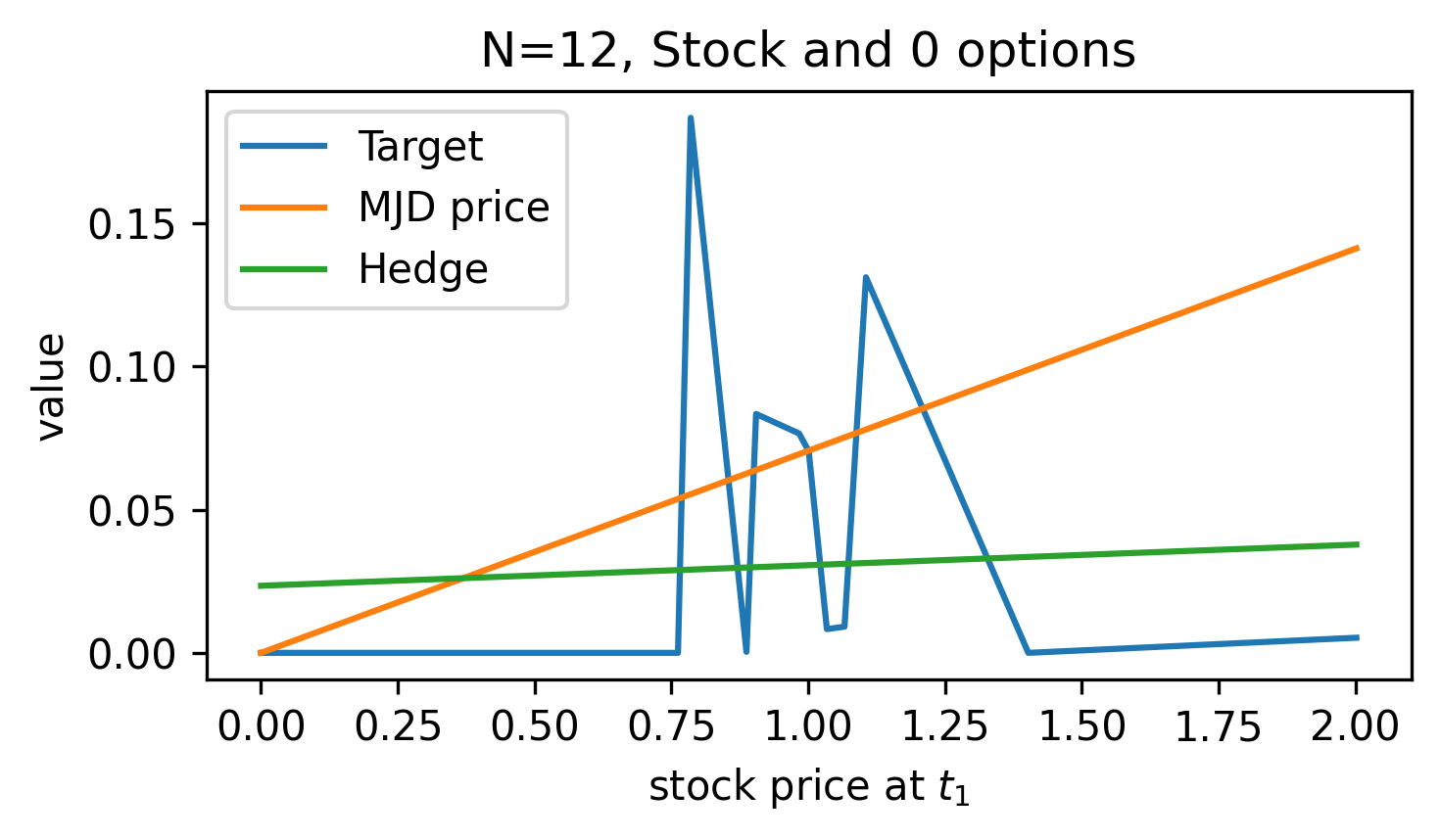}
         \caption{Plot for only stock }
         \label{fig:0option_condplot_mjd_forward}
     \end{subfigure}
     \hfill
     \begin{subfigure}[b]{0.4\textwidth}
         \centering
         \includegraphics[width=\textwidth]{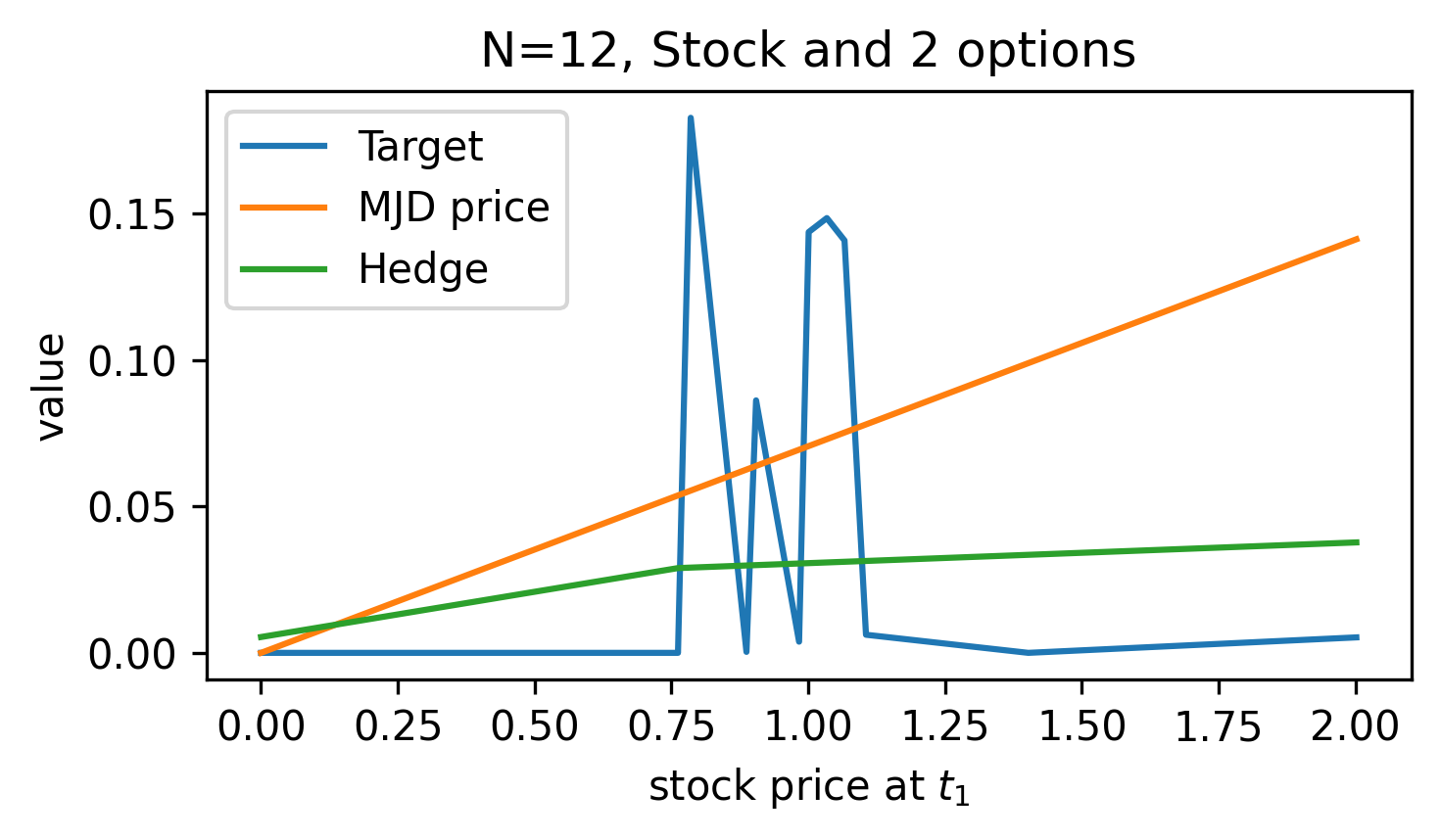}
         \caption{Plot for stock and 2 options}
         \label{fig:2option_condplot_mjd_forward}
     \end{subfigure}
     \hfill
     \begin{subfigure}[b]{0.4\textwidth}
         \centering
         \includegraphics[width=\textwidth]{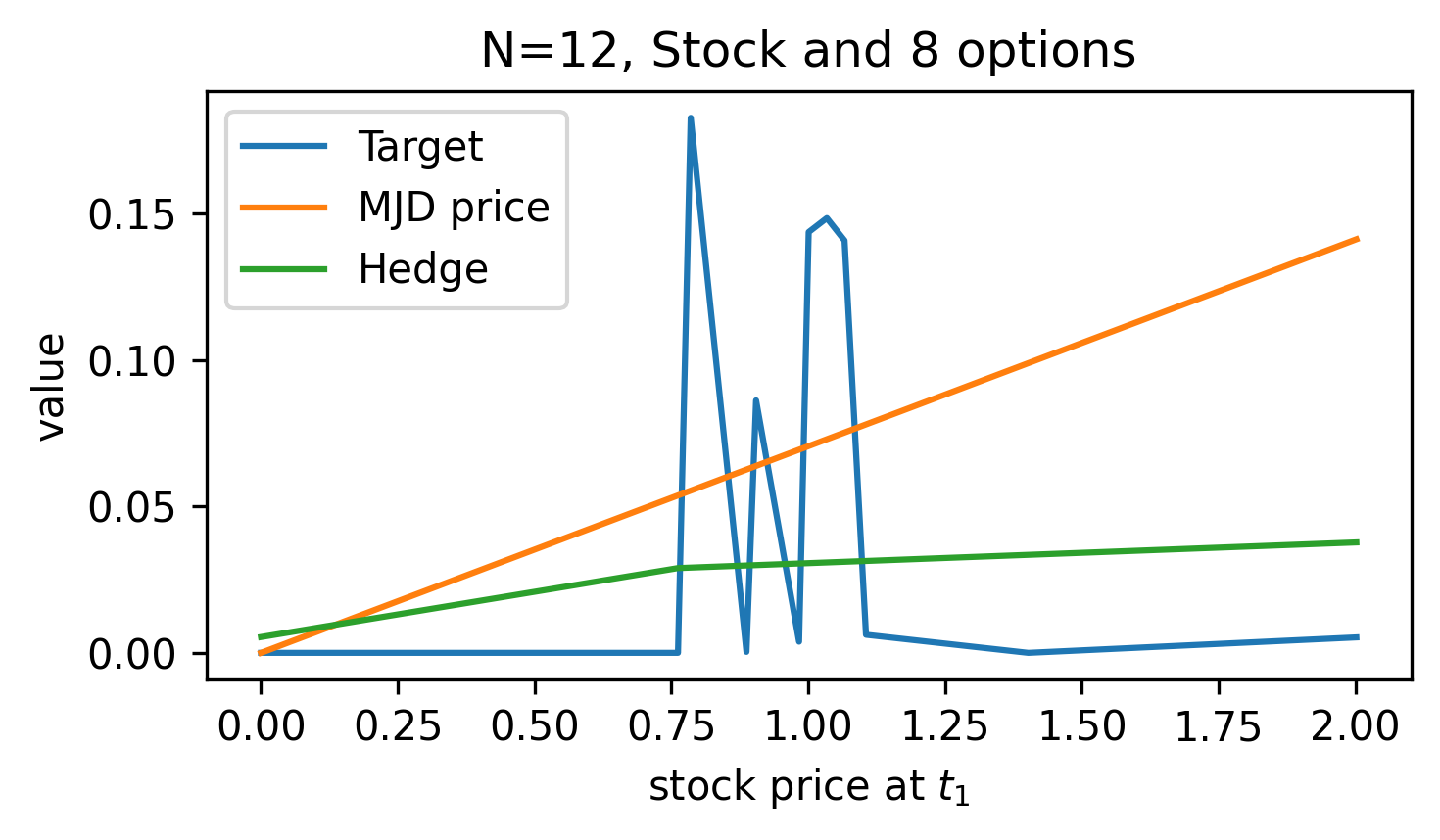}
         \caption{Plot for stock and 8 options}
         \label{fig:8option_condplot_mjd_forward}
     \end{subfigure}
     \hfill
     \begin{subfigure}[b]{0.4\textwidth}
         \centering
         \includegraphics[width=\textwidth]{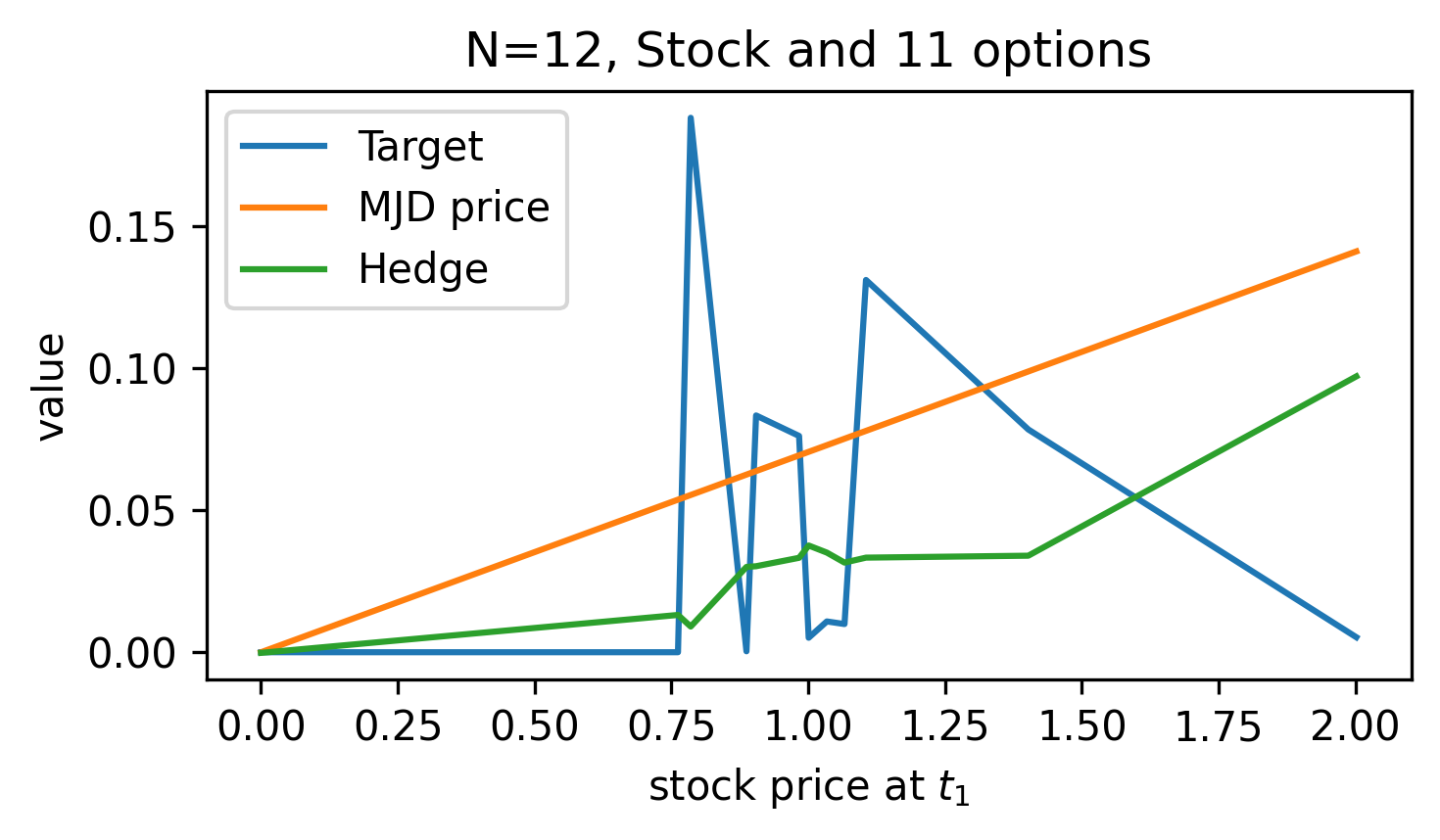}
         \caption{Plot for stock and 11 options}
         \label{fig:11option_condplot_mjd_forward}
     \end{subfigure}
        \caption{Plots of the conditional value of the target option under the worst possible scenarios and the corresponding hedging portfolio values for an increasing number of options for the Forward start option under the Merton Jump Diffusion model.}
        \label{fig:cond_plots_MJDForward}
\end{figure}
The plots in Figure \ref{fig:cond_plots_MJDForward} show a similar situation as in the BS model, where neither the hedging portfolio (green line) nor the target option price obtained using the worst case probabilities (blue line) aligns with the true price given by the original line. However, the addition of more options ( $\geq8$)  yields a better fit for the target option to the true price than with fewer options ($\leq8)$. 

We can also observe that the plot for subplot (\ref{fig:8option_condplot_mjd_forward}) for a portfolio with $8$ options and stock is identical to plot for \ref{fig:2option_condplot_mjd_forward}) for the portfolio with $2$ options and stock in Figure \ref{fig:cond_plots_MJDForward} indicating that the addition of more options does not necessarily reduce the absolute hedging error of the portfolio for the given choices of the strikes.

Next, we focus our study on the performance of the hedging portfolios generated using the experiments in this section up to the short maturity $t_1$ of the options in the hedging portfolio. 

\section{Results from simulation}
 In this section, we utilize the optimal weights obtained from solving the min-max optimization problem to compute the error statistics of the resulting hedging portfolio at the short maturity $t_1$.
 
 To calculate the Peak Potential Future Exposure (PFE), we simulate $10^4$ stock paths at each of the equispaced time points over the time interval $[0,t_1]$ with a spacing of $h=0.1$. 
 
 The peak $99\textsuperscript{th}$ and $95\textsuperscript{th}$ PFEs are calculated by taking the maximum of the $99\textsuperscript{th}$ and $95\textsuperscript{th}$ percentiles of the hedging error over the equispaced time points. The peak $5\textsuperscript{th}$ and $1\textsuperscript{st}$ PFEs are computed similarly by taking the minima. The hedging error at any time $t\in[0,t_1]$ is given by
\begin{align}
\begin{split}
   &\text{Hedging error at time $t$} \\
    &= \text{Value of target option at time $t$}
    - \text{Value of the hedging portfolio at time $t$}. 
\end{split}
\end{align}
\subsection{Black Scholes Model}
We use the same parameters as used in Section \ref{subsec:BS_exp} to calculate the necessary statistics for the Asian option and the Forward start option.
\subsubsection{Asian Option}

Figure \ref{fig:minmaxvsmae_bsasian} shows the value of the objective function obtained by solving the min-max problem for an increasing number of options in the hedging portfolio. The orange line represents the Mean Absolute Error (MAE) of the hedging portfolio at the short maturity $t_1$, obtained using the stock price simulations. We can observe a sharp decline in both the values as we increase the number of options beyond $5$. The min-max objective gives an upper bound for the mean absolute error.

Figure \ref{fig:peak_pfe_bsasian} gives the peak PFEs of the hedging error for an increasing number of options. We obtain a similar conclusion that there is a significant decrease in the $99\textsuperscript{th}$ and $95\textsuperscript{th}$ percentiles beyond $5$ options in the hedging portfolio, with both the values becoming almost identical beyond $8$ options.

\begin{figure}
    \centering
    \includegraphics[width=0.5\linewidth]{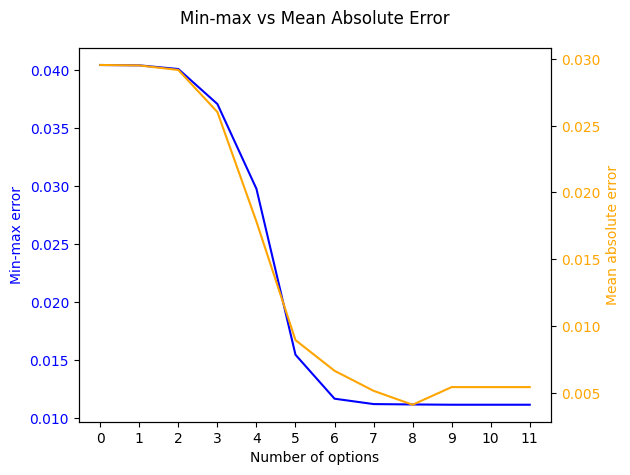}
    \caption{Min-Max versus Mean Absolute Error plot for the Asian option under the Black Scholes model.}
    \label{fig:minmaxvsmae_bsasian}
\end{figure}
\begin{figure}
     \centering
     \begin{subfigure}[b]{0.4\textwidth}
         \centering
         \includegraphics[width=\textwidth]{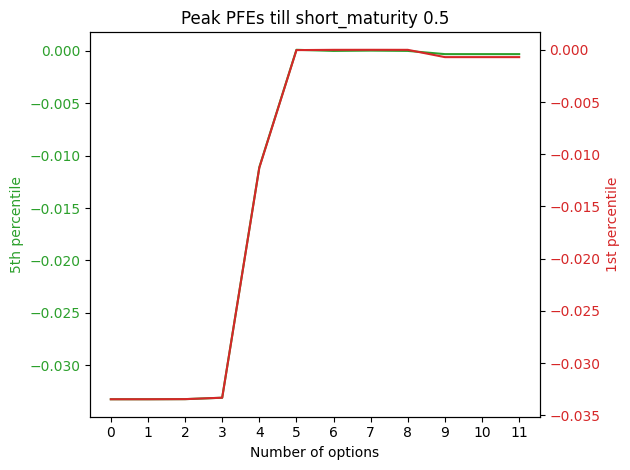}
         \caption{Peak $5\textsuperscript{th}$ and $1\textsuperscript{st}$ PFE }
         \label{fig:5th_and_1st_pfe_bs_asian}
     \end{subfigure}
     \hfill
     \begin{subfigure}[b]{0.4\textwidth}
         \centering
         \includegraphics[width=\textwidth]{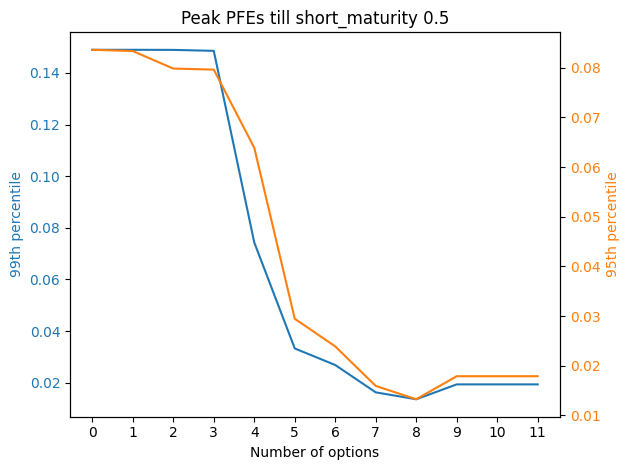}
         \caption{Peak $99\textsuperscript{th}$ and $95\textsuperscript{th}$ PFE }
         \label{fig:99th_and_95th_pfe_bs_asian}
     \end{subfigure}
      \caption{Peak PFE plot for the Asian option under the Black Scholes model.}
    \label{fig:peak_pfe_bsasian}
        
\end{figure}

\subsubsection{Forward Start Option}\label{subsec:forward_BS_Tmaturity_simulations}
Figure \ref{fig:minmaxvsmae_bsforward} shows that the min-max and MAE are marginally reduced by the addition of more options. The min-max error in this case again serves as an upper bound to MAE. 

Figure \ref{fig:peak_pfe_bsforward}, on the contrary, shows a sharp increase in the $99\textsuperscript{th}$ and $95\textsuperscript{th}$ peak PFEs on the addition of more options but a considerable drop in the corresponding $5\textsuperscript{th}$ and $1\textsuperscript{st}$ peak PFEs. \textbf{This indicates an important fact: while the weights we obtain heuristically by solving (\ref{ourLP}) may help minimize the error in the worst possible scenario, they need not be the set of weights corresponding to the least hedging error under the true price dynamics.}

\begin{figure}
    \centering
    \includegraphics[width=0.5\linewidth]{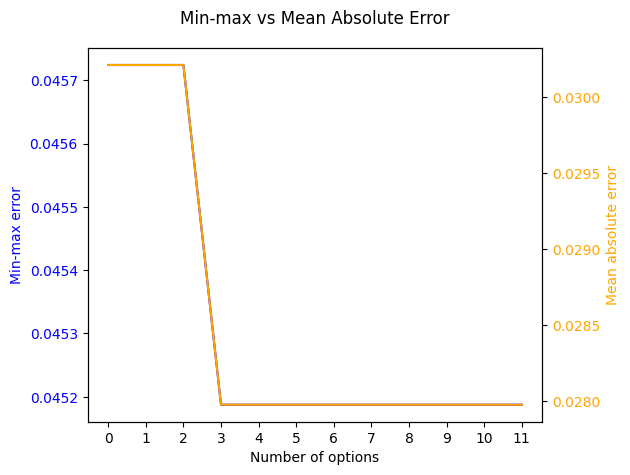}
    \caption{Min-Max versus Mean Absolute Error plot for the Forward start option under the Black Scholes model.}
    \label{fig:minmaxvsmae_bsforward}
\end{figure}
\begin{figure}
    \centering
    \begin{subfigure}[b]{0.4\textwidth}
         \centering
         \includegraphics[width=\textwidth]{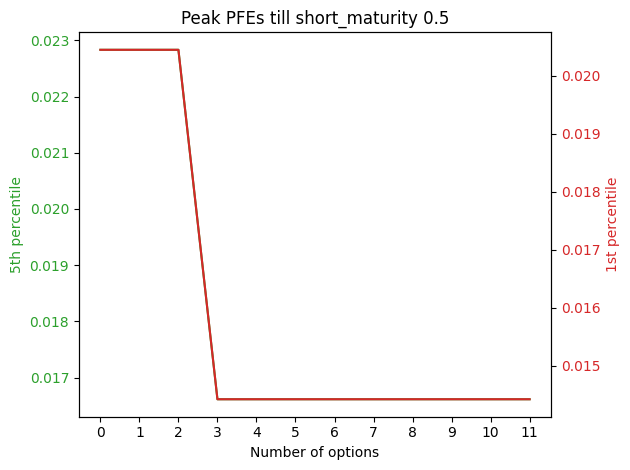}
         \caption{Peak $5\textsuperscript{th}$ and $1\textsuperscript{st}$ PFE }
         \label{fig:5th_and_1st_pfe_bs_forward}
     \end{subfigure}
     \hfill
     \begin{subfigure}[b]{0.4\textwidth}
         \centering
         \includegraphics[width=\textwidth]{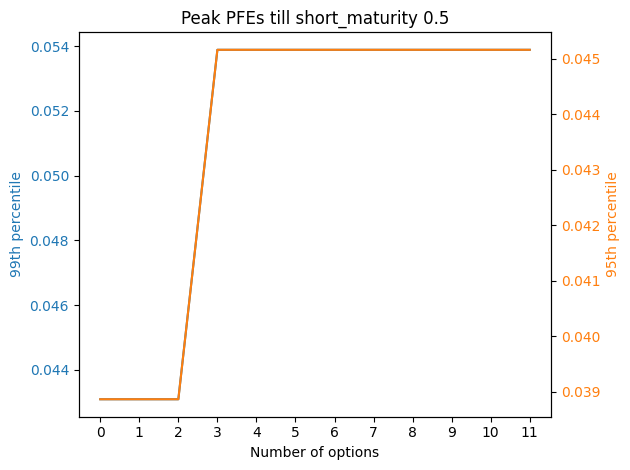}
         \caption{Peak $99\textsuperscript{th}$ and $95\textsuperscript{th}$ PFE }
         \label{fig:99th_and_95th_pfe_bs_forward}
     \end{subfigure}
    \caption{Peak PFE plot for the Forward start option under the Black Scholes model.}
    \label{fig:peak_pfe_bsforward}
\end{figure}

\textbf{Effect of including options with maturity $T$ in the hedging portfolio:}
As an additional experiment, we now include the options with maturity $T$ and let the hedging portfolio comprise cash, stock, and options corresponding to all  $22$ strikes with maturity $t_1$ and $T$ used to calculate the marginal distributions. We compute the weights for the min-max problem in this case at time $t_1$ and compare the corresponding statistics with the performance of the super-hedging portfolio obtained by solving the dual problem to the original MOT option pricing problem.

Figure \ref{fig:stats_dual_versus_minmax} displays the peak PFEs and Mean Absolute Errors of the hedging portfolio obtained using the dual and the min-max problems, respectively. It can be observed that while the $99\textsuperscript{th}$ and $95\textsuperscript{th}$ peak PFEs for the min-max solution are higher than that of the dual solution, the mean absolute errors for both cases are similar. The exact values of the MAE for the min-max hedging portfolio and dual hedging portfolio are $ 0.0126$ and $0.0120$, respectively. The peak $5\textsuperscript{th}$ and $1\textsuperscript{st}$ peak PFEs for both the hedging portfolios are negative, indicating that the hedging portfolios are higher in value than the target option in these scenarios, which is profitable from the option writer's perspective.
\begin{figure}
     \centering
         \centering
         \includegraphics[width=\textwidth]{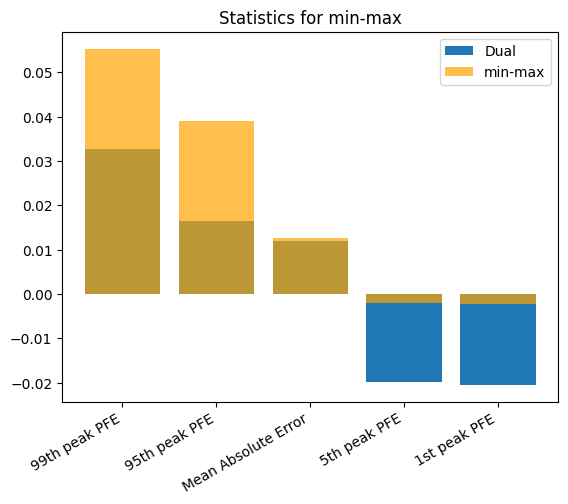}
    
        \caption{Values of the Peak PFEs and the Mean Absolute Errors for the hedging portfolios with options of both maturities $t_1$ and $T$ obtained using the dual and the min-max problems for the Forward option under the Black Scholes model.}
        \label{fig:stats_dual_versus_minmax}
\end{figure}
\begin{table}
    \centering
    \begin{tabular}{|c|c|c|c|}
    \hline
        &$MM_{t_1}$ & $MM_{T}$ & $\text{Dual}_{SH}$    \\
        \hline
       \text{Max}& 0.0452 &  0.0223 & 0.0403\\
       $MAE$ & 0.0280 &  0.0126 & 0.0120\\
         \hline
    \end{tabular}
    \caption{Worst possible errors for the three hedging portfolios for the Forward start option under the Black Scholes model }
     \label{tab:min-maxversusdualBS_Forward}
\end{table}
\subsection{Merton Jump Diffusion Model}
\subsubsection{Forward Start Option}
Figure \ref{fig:minmaxvsmae_MJD_Forward} displays the peak PFEs and Mean Absolute Errors of the hedging portfolio obtained using the dual and the min-max problems, respectively, for an increasing number of options. We observe that the min-max error (denoted by the blue line) is higher than the MAE until $10$ options in the hedging portfolio, providing an upper bound to the MAE as desired.

Figure \ref{fig:peak_pfe_mjdforward} gives the corresponding peak PFEs. The results, especially the $5\textsuperscript{th}$ and $1\textsuperscript{st}$ peak PFEs, again indicate that the hedging portfolio with an increasing number of options corresponding to the solution of the min-max problem need not be the one that reduces the error under the true dynamics. The resulting portfolio would instead minimize the error in the worst possible scenario.

 \footnote{We would also like to highlight an important point that solutions obtained using the min-max algorithm are subject to numerical instabilities based on the choice of the parameters and the type of the optimization algorithm involved. This is beyond the scope of our study and hence we do not discuss it. } 
\begin{figure}
    \centering
    \includegraphics[width=0.5\linewidth]{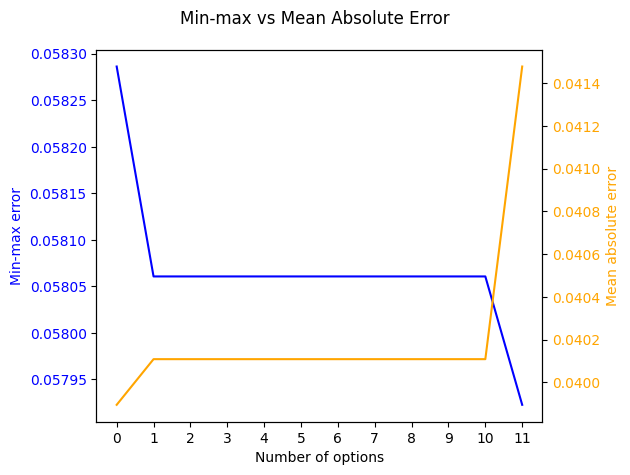}
    \caption{Min-Max versus Mean Absolute Error plot for the Forward option under the Merton jump Diffusion model.}
    \label{fig:minmaxvsmae_MJD_Forward}
\end{figure}
\begin{figure}
     \centering
     \begin{subfigure}[b]{0.4\textwidth}
         \centering
         \includegraphics[width=\textwidth]{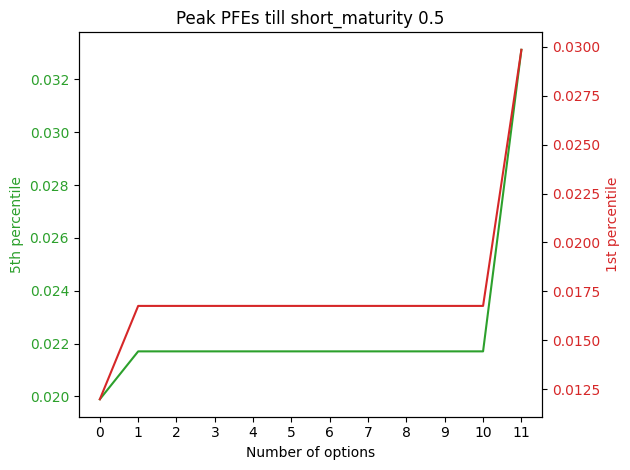}
         \caption{Peak $5\textsuperscript{th}$ and  $1\textsuperscript{st}$ PFE }
         \label{fig:5th_and_1st_pfe_mjd_forward}
     \end{subfigure}
     \hfill
     \begin{subfigure}[b]{0.4\textwidth}
         \centering
         \includegraphics[width=\textwidth]{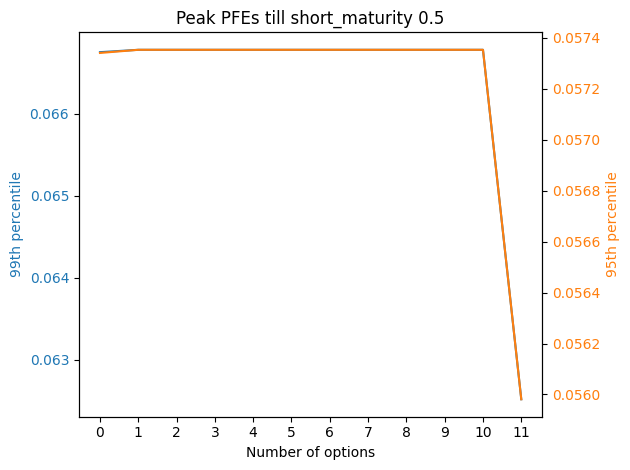}
         \caption{Peak $99\textsuperscript{th}$ and $95\textsuperscript{th}$ PFE }
         \label{fig:99th_and_95th_pfe_mjd_forward}
     \end{subfigure}
      \caption{Peak PFE plot for the Forward start option under the Merton jump diffusion model.}
    \label{fig:peak_pfe_mjdforward}
        
\end{figure}
\section{Conclusion}\label{Sec:Conclusion}
In this paper, we present a model-free approach to hedging options of maturity $T$ whose payoffs depend on the value of an asset at two distinct time points. The hedging portfolio comprises a cash position, the underlying asset, and plain-vanilla options on the same underlying asset of short maturity $t_1$, where $0<t_1<T$. 

We formulate the problem of worst-case absolute hedging errors at the maturity of the short-term maturity options as a min-max optimization problem. The data for the problem consists of a finite number of liquid plain vanilla options at the two maturity points, which allow us to approximate the true marginal distributions. 

The inner maximization problem turns out to be a modified Martingale Optimal Transport problem. The worst-case error refers to the maximization over all martingale measures with the given marginals. A solution to this problem yields a cost-effective portfolio that minimizes the worst-case error at short maturity.

We also derive a theoretical upper bound on the absolute hedging error at the longer maturity $T$ due to the availability of finitely many equally spaced strikes over a bounded support of the true underlying measure.

A wide range of numerical examples, including Asian and Forward start options, under the Black-Scholes and Merton's Jump diffusion model, illustrate the utility of this method. The experiments highlight an important fact: \textbf{the traditional super hedge, while being more expensive than our hedging portfolio, does not necessarily yield the minimum possible worst-case error.} 

During our experiments, we observed that the numerical solutions obtained by solving the min-max problem in the discrete case may not be unique and depend on the choice of the solver. We have used a Gurobi optimization solver for the inner maximization problem and utilised the SLSQP method of scipy.optimize library for the outer minimization problem.

\section{References}
\bibliographystyle{plain} 
\bibliography{references}

\end{document}